\pdfoutput=1
\documentclass[
  superscriptaddress,
  amsmath,amssymb,
  aps,
  prx,
  reprint,
  floatfix,
  raggedbottom,
]{revtex4-2}

\usepackage[utf8]{inputenc}
\usepackage[T1]{fontenc}
\usepackage{graphicx}
\usepackage{bm}
\usepackage{mathtools}
\usepackage{amsthm}
\usepackage{needspace}
\usepackage[colorlinks,citecolor=blue,urlcolor=blue,linkcolor=blue]{hyperref}

\hypersetup{
  pdftitle={Oscillator–Qubit Primitives for a Molecular Quantum Dynamics Simulator},
  pdfauthor={Jungsoo Hong; Joonsuk Huh},
  pdfsubject={Coherent synthesis of nonlinear multimode phase gates on oscillator-qubit processors},
  pdfkeywords={molecular quantum dynamics, non-Born-Oppenheimer dynamics, oscillator-qubit processors, quantum signal processing, nonlinear phase compilation}
}

\theoremstyle{plain}
\newtheorem{theorem}{Theorem}

\newtheorem{lemma}[theorem]{Lemma}

\theoremstyle{definition}

\newcommand{\ii}{\mathrm{i}}
\newcommand{\ee}{\mathrm{e}}
\newcommand{\bbR}{\mathbb{R}}
\newcommand{\bQ}{\hat{\bm Q}}
\newcommand{\CDx}[1]{\operatorname{CD}_{x}\!\left(#1\right)}
\newcommand{\Rz}{R_{\rm z}}
\newcommand{\Rx}{R_{\rm x}}
\newcommand{\Nel}{N_{\mathrm{el}}}
\newcommand{\calO}{\mathcal{O}}

\newcommand{\norm}[1]{\left\lVert #1\right\rVert}
\newcommand{\abs}[1]{\left\lvert #1\right\rvert}
\newcommand{\ket}[1]{\left|#1\right\rangle}
\newcommand{\bra}[1]{\left\langle#1\right|}
\newcommand{\bmu}{\bm{\mu}}
\newcommand{\PiN}{\Pi_{\bm n^{\max}}}
\newcommand{\eps}{\varepsilon}

\begin{document}

\title{Oscillator--Qubit Primitives for a Molecular Quantum Dynamics Simulator}

\author{Jungsoo Hong}
\affiliation{SKKU Advanced Institute of Nano Technology (SAINT), Sungkyunkwan University, Suwon 16419, Republic of Korea}

\author{Joonsuk Huh}
\email{joonsukhuh@yonsei.ac.kr}
\affiliation{Department of Chemistry, Yonsei University, Seoul 03722, Republic of Korea}
\affiliation{Department of Quantum Information, Yonsei University, Incheon 21983, Republic of Korea}
\affiliation{Department of Computational Science and Engineering, Yonsei University, Seoul 03722, Republic of Korea}

\date{September 30, 2026}

\begin{abstract}
Molecular quantum dynamics simulations that treat both electrons and nuclei quantum mechanically are crucial for predicting chemical reactions. With classical computation, a full wave-function representation of both types of particles requires resources that grow exponentially with their number. Oscillator--qubit processors, including trapped-ion and circuit QED devices, represent electronic states with qubits and nuclear vibrations with oscillators. Their native operations produce linear vibronic interactions, from which programmable anharmonic interactions can be built. Product-formula approaches build the corresponding phase gates by repeating small noncommuting phase steps, and quantum signal processing (QSP) approaches approximate each gate directly with a Fourier series. Simulations apply many such gates in sequence, so accumulated cost and error limit how long the dynamics can be followed. We propose near-optimal quantum primitives for molecular quantum dynamics simulation, covering the representations of both types of particles and their interactions. We synthesize nonlinear and multimode bosonic phase gates using generalized Jacobi--Anger (GJA) expansions within QSP, with a cost that grows near-optimally in both interaction strength and precision. Our interaction model and gate synthesis each have their own accuracy setting. Chosen together, these settings make our multimode benchmark circuits about one-third shorter than those of the previous direct QSP approach. We then simulate energy exchange between the stretching and bending vibrations of a simplified carbon dioxide model. With circuit-QED noise, a short simulation narrowly misses our accuracy criteria at qubit and oscillator lifetimes near $2$~ms and $5$~ms. These lifetimes mark a concrete checkpoint on the way to simulating anharmonic multimode molecular quantum dynamics on near-term oscillator--qubit processors.
\end{abstract}


\maketitle

\section{Introduction}
\label{sec:intro}

\begin{figure*}[!t]
\centering
\includegraphics[width=0.98\textwidth]{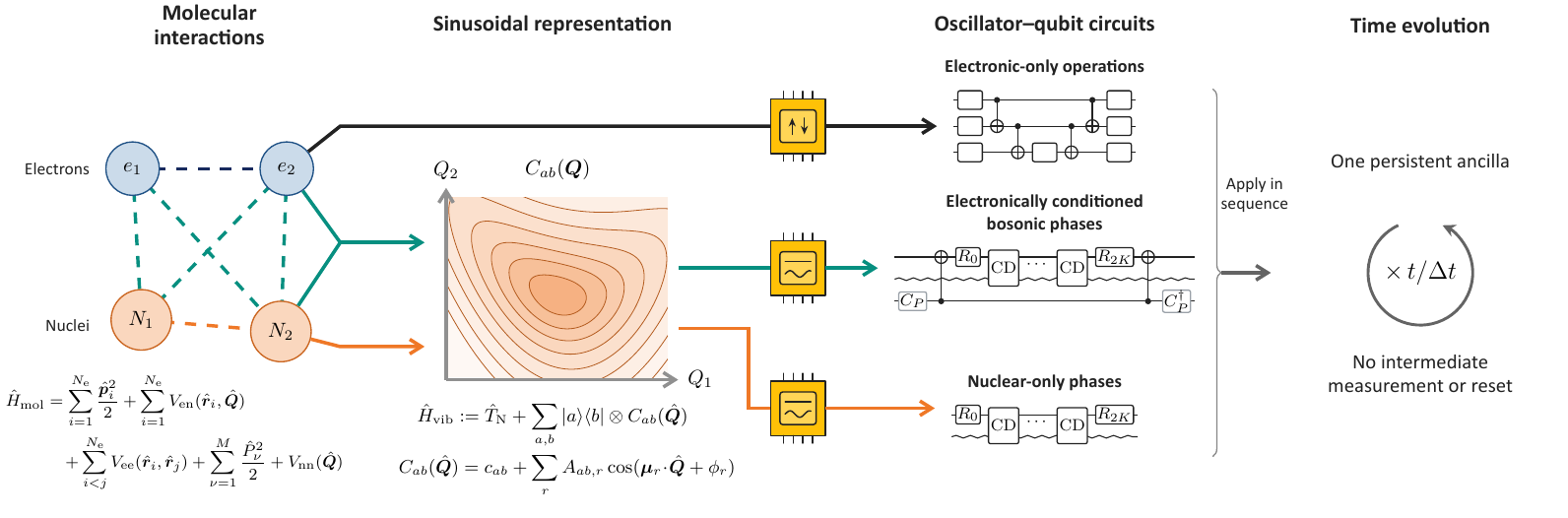}
\caption{From molecular interactions to oscillator--qubit circuits. The interactions among electrons and nuclei (left) are expressed in an electronic basis, where each matrix element $C_{ab}(\bQ)$ is a sum of sinusoidal functions of the nuclear coordinates (center). The contour plot sketches one such element, and the colored arrows indicate the origin of each term only schematically, because the electronic basis mixes the microscopic interactions. Each term is then compiled into an oscillator--qubit circuit (right). Terms without coordinate dependence become electronic-only operations, and coordinate-dependent terms become bosonic phase gates that are either conditioned on the electronic state or act on the nuclei alone. Terms that meet the grouping conditions of Sec.~\ref{sec:common-architecture} share one circuit. Within a time step the circuits are applied in sequence, and the step is repeated $t/\Delta t$ times with one ancilla that is kept throughout without intermediate measurement or reset.}
\label{fig:schematic}
\end{figure*}

When nuclear motion brings electronic states close in energy, as in light-induced reactions, transitions between these states can redirect the reaction pathway~\cite{Yarkony1996CI,CrespoOteroBarbatti2018NAMQC}. These dynamics couple electronic populations to nuclear wave packets, whose motion involves energy redistribution among vibrational modes as well as quantum effects such as tunneling and interference~\cite{HaMacDonell2025}. A faithful description must then keep the nuclear wave function together with the electronic states.

Classical trajectory methods avoid the full nuclear wave function by replacing nuclear motion with classical trajectories and evolving electronic amplitudes along them~\cite{Tully1990SurfaceHopping,LiTullySchlegelFrisch2005Ehrenfest}. Semiclassical methods recover part of the nuclear quantum behavior through interference between trajectories~\cite{Miller2001SCIVR,BenNun2000AIMS}, but these approximations can still miss tunneling or interference that controls an observable. The multiconfiguration time-dependent Hartree (MCTDH) method and its multilayer extension (ML-MCTDH) offer another route by propagating a compressed wave function~\cite{BeckMeyer2000,WangThoss2003MLMCTDH,Manthe2008MLMCTDH}. This compression makes larger quantum calculations accessible, although strong correlations between modes can limit its effectiveness and make propagation expensive.

Oscillator--qubit (OQ) processors encode nuclear motion in bosonic modes and electronic states in qubits~\cite{Feynman1982,Lloyd1996,LiuHybridOQ2026,HaMacDonell2025}, with realizations in trapped ions~\cite{McGarry2026AnharmonicDynamics,Chalermpusitarak2025TGIFS,Rainaldi2026TrigBenchmark} and circuit QED devices~\cite{Diringer2024CNOD,Eickbusch2022}. These processors hold the molecular wave function without explicitly storing its exponentially many amplitudes. The remaining task is to evolve this state under the molecular Hamiltonian with the operations that the processor provides. The native operations implement harmonic vibrations and electron--nuclear couplings that are linear in the nuclear coordinates~\cite{LiuHybridOQ2026}. Anharmonic couplings, which drive vibrational energy redistribution and reaction dynamics, are programmed through nonlinear phase gates built from native conditional displacements and ancilla-qubit rotations (Fig.~\ref{fig:schematic}).

On an OQ processor, an anharmonic interaction that depends on the nuclear coordinates is implemented as a bosonic phase gate whose phase is a nonlinear function of these coordinates. Many theoretical and experimental studies of such gates have appeared recently. Theory has proposed several approaches: numerically optimized Fourier synthesis~\cite{ParkFilip2024}, trigonometric-gate-implemented Fourier synthesis (TGIFS)~\cite{Chalermpusitarak2025TGIFS,Rainaldi2026TrigSineGordon}, and quantum signal processing (QSP) on OQ processors~\cite{LiuMixedADQSP2024,HongOQGQSP2026}. Experiments have begun to realize such gates through bosonic QSP and related Fourier synthesis in trapped ions and circuit QED~\cite{Lee2026TrappedIonCubic,Fontaine2026Programming}. On photonic processors, QSP has been used to synthesize nonlinear functions~\cite{Karooby2026PhotonicQSP}, and a two-mode linear-optical architecture has been proposed for arbitrary finite Fourier series~\cite{Mauser2026OpticalFourier}.

Two of the proposed approaches have complementary strengths. TGIFS represents the interaction by sinusoidal functions of the nuclear coordinates, a form matched to the native conditional displacements~\cite{Chalermpusitarak2025TGIFS,McGarry2026AnharmonicDynamics,Rainaldi2026TrigBenchmark}. It builds each phase gate as a product formula that repeats small phase steps, so stronger phases and tighter tolerances require longer circuits. Oscillator--qubit generalized QSP (OQ--GQSP) instead synthesizes the whole phase in one circuit~\cite{HongOQGQSP2026}. It approximates the unitary phase response directly, with a single approximation for the interaction and the gate.

Our compiler combines the strengths of both: it takes the sinusoidal model of the potential from TGIFS and replaces the repeated steps with one QSP circuit for the whole phase. The generalized Jacobi--Anger (GJA) expansion supplies the response coefficients of this circuit explicitly. We call this route Hamiltonian-first, because the potential is approximated before the gate is synthesized. Since an error in the sinusoidal model changes only the phase of the gate, it can be bounded separately from the ancilla leakage introduced during synthesis. The model and its error bound are obtained once and serve every evolution time, whereas the circuit is synthesized again for each requested time and accuracy (Sec.~\ref{sec:hamiltonian-first}). Appendix~\ref{app:working-space} gives the domain and representation conditions.

When a potential couples several modes, its sinusoidal model can contain products of sinusoidal functions of different coordinates. Product-to-sum identities rewrite each product as a sum of terms, each a single sinusoidal function of one linear combination of the coordinates. This rewriting keeps the coupling between the modes. Terms can be merged into a single phase gate when they share the same combination of coordinates and the same electronic Pauli operator. Their frequencies must also be integer multiples of one base frequency. Our compiler, GJA--OQ--QSP, builds one QSP circuit for each such group of terms and applies the circuits one after another with one ancilla, without intermediate measurement or reset. This sequence can be reused at every step of an outer evolution algorithm (Secs.~\ref{sec:controlled-lift} and \ref{sec:composition}).

The construction also builds on the single-tone Jacobi--Anger (JA) expansion used in QSP-based Hamiltonian simulation~\cite{GilyenSuLowWiebe2019,MotlaghWiebe2024} and on the Fourier coefficient bound of OQ--GQSP~\cite{HongOQGQSP2026}. It adds four elements. First, every synthesized gate is completed to a unitary operation on the ancilla and the oscillators together, so successive gates share one ancilla without postselection. Earlier oscillator--qubit circuits postselect the ancilla, either once after the sequence~\cite{Chalermpusitarak2025TGIFS,McGarry2026AnharmonicDynamics,Rainaldi2026TrigBenchmark} or after each gate~\cite{HongOQGQSP2026}. Second, all terms of a group are compiled into one circuit, whose length is governed by the largest gradient of their combined phase. Third, we prove that the summed displacement magnitude of such a circuit is at least the largest momentum impulse that the target gate imparts. The grouped circuit attains this bound to leading order. A related derivative bound limits the Rabi time of polynomial phase gates on low-energy inputs~\cite{Yang2026RabiPolyPhase}, whereas our statement concerns sinusoidal potentials with the error bounded at every coordinate value. Fourth, the completion leaves one phase free, and its alternation between successive time steps cancels part of the coherent error that passes between gates through the shared ancilla. For a single sinusoidal term, the displacement count is optimal up to a constant factor in phase strength and precision, for exact displacements in the circuit models of Sec.~\ref{sec:complexity-results}.

The numerical tests compare conditional-displacement (CD) costs for potentials with one and with several sinusoidal terms under common error criteria. To test the shared ancilla over many gates, we follow energy exchange between the stretching and bending vibrations of a simplified carbon dioxide model (Sec.~\ref{sec:benchmarks}). This model involves nuclear motion only, and we compare the compiled evolution with a reference. We then follow a shorter window under pulse-level noise.

Section~\ref{sec:problem-setting} connects the molecular representation and native displacements through the phase-gate displacement series. Section~\ref{sec:common-architecture} develops serial GJA compilation, QSP completion, and coherent electronic control. Section~\ref{sec:complexity-results} develops the query and action guarantees, compares potential and gate approximation, and examines coherent error allocation. Section~\ref{sec:benchmarks} presents circuit costs and coherent and noisy dynamics, followed by the discussion in Sec.~\ref{sec:conclusion}. The proofs are collected in Appendices~\ref{smsec:ja-details} and \ref{app:proofs}, with circuit conventions in Appendix~\ref{smsec:circuit-details} and numerical validation in Appendix~\ref{app:numerical-validation}.

\section{Molecular representations and native signals}
\label{sec:problem-setting}

\subsection{Hamiltonian representation and gate-synthesis targets}
\label{sec:hm-target}

The molecular representation determines the coordinate-dependent gates required by the outer time-evolution method. Classical preprocessing first chooses the nuclear coordinates and the electronic basis. These choices define the coordinate-dependent vibronic model in scaled coordinates with $\hbar=1$:
\begin{equation}
\begin{aligned}
\hat H_{\rm mol}
&=\sum_{i=1}^{N_{\rm e}}\frac{\hat{\bm p}_i^2}{2}
 +\sum_{i=1}^{N_{\rm e}}V_{\rm en}(\hat{\bm r}_i,\bQ)\\
&\quad{}+\sum_{i<j}^{N_{\rm e}}V_{\rm ee}(\hat{\bm r}_i,\hat{\bm r}_j)
 +\sum_{\nu=1}^{M}\frac{\hat P_\nu^2}{2}
 +V_{\rm nn}(\bQ)\\
&\longmapsto\hat H_{\rm vib}:=\hat T_{\rm N}
 +\sum_{a,b}\ket a\bra b\otimes C_{ab}(\bQ).
\end{aligned}
\label{eq:molecular-vibronic}
\end{equation}
Here the potentials $V_{\rm en}$, $V_{\rm ee}$, and $V_{\rm nn}$ denote the electron--nuclear, electron--electron, and nuclear--nuclear interactions. The sums run over $N_{\rm e}$ electrons with positions $\hat{\bm r}_i$ and momenta $\hat{\bm p}_i$. The vector $\bQ=(\hat Q_1,\ldots,\hat Q_M)$ contains the dimensionless nuclear coordinates $\hat Q_j=(\hat a_j+\hat a_j^\dagger)/\sqrt2$, where $\hat a_j$ is the annihilation operator of mode $j$. The nuclear kinetic energy is $\hat T_{\rm N}=\sum_\nu\hat P_\nu^2/2$. In a finite diabatic or quasi-diabatic electronic basis, $C_{aa}$ are potentials and $C_{ab}$ with $a\ne b$ are electronic-state couplings~\cite{MacDonell2021AnalogDynamics}.

Within the retained vibronic model, we parameterize the matrix elements on the declared nuclear-coordinate domain as
\begin{equation}
C_{ab}(\bQ)=c_{ab}+\sum_r A_{ab,r}\cos(\bmu_r\cdot\bQ+\phi_r).
\label{eq:vibronic-sinusoidal}
\end{equation}
The matrices $c=(c_{ab})$ and $A_r=(A_{ab,r})$ are Hermitian, which ensures $C_{ba}=C_{ab}^{*}$ even for complex couplings. The frequencies $\bmu_r$ and phases $\phi_r$ are real, repeated frequencies may carry different phases, and $\phi_r=-\pi/2$ gives a sine term. The matrix coefficients can be fitted directly to potential and coupling data in a consistent electronic basis, without an intermediate polynomial model. The finite frequency set and the fitting error belong to the supplied representation.

Qubit encoding and Pauli expansion then give the retained compiler input
\begin{equation}
\hat H_{\rm in}:=\hat H_{\rm G}
 +\sum_{\ell=1}^{L}P_\ell\otimes f_\ell(\bQ).
\label{eq:cmqb-general}
\end{equation}
The real coefficients $f_\ell$ inherit the same cosine representation. The term $\hat H_{\rm G}$ collects qubit-only and Gaussian contributions, including conditional-Gaussian terms supported by the gate model. The remaining terms are of two kinds. Terms with $P_\ell=I$ act identically on all electronic states. Terms with $P_\ell\ne I$ describe state-dependent potentials and transitions. The latter include $X$- and $Y$-type components. The identity component can include projected electronic energies as well as the bare nuclear--nuclear interaction. The microscopic sectors in Fig.~\ref{fig:schematic} therefore do not map one-to-one onto these Pauli sectors. Appendix~\ref{app:working-space} gives alternative representations and accounts for their discrepancy from a reference molecular model.

For each term, the outer time-evolution method supplies a time weight $\tau$.  The gate-synthesis task is

\begin{equation}
U_{\ell,f}(\tau)=\exp[-\ii\tau P_\ell\otimes f_\ell(\hat{\bm Q})].
\label{eq:Pauli-controlled-phase}
\end{equation}

For a diagonal $P_\ell$, this gate conditions the bosonic phase on the electronic state. A nondiagonal string can be rotated to a diagonal Pauli string and rotated back afterward. The compiler implements each such coordinate-phase call. The outer method determines how calls with noncommuting generators are combined.

For a single sinusoidal component, the target gate is
\begin{equation}
U_P(\Lambda;\bmu,\phi)
=\exp[-\ii\Lambda P\otimes\cos(\bmu\cdot\hat{\bm Q}+\phi)].
\label{eq:pauli-ridge-target}
\end{equation}
Here $\bmu$ selects a linear combination of the nuclear coordinates, and $\phi$ is an offset. This dependence on one collective coordinate defines a sinusoidal ridge. The electronic Pauli string $P$ specifies the electronic conditioning. The choice $P=I$ gives nuclear-only evolution. The phase amplitude $\Lambda=A\Delta t$ combines the coupling energy $A$ and time step $\Delta t$.

Products of sinusoidal functions also admit a ridge representation. For example,
\begin{equation}
\begin{aligned}
\cos\hat Q_1\,\cos\hat Q_2
&=\tfrac12\cos(\hat Q_1+\hat Q_2)\\
&\quad+\tfrac12\cos(\hat Q_1-\hat Q_2).
\end{aligned}
\label{eq:product-to-ridge-example}
\end{equation}
The two ridges depend on the collective coordinates $\hat Q_1+\hat Q_2$ and $\hat Q_1-\hat Q_2$, and each acts jointly on both modes. Their generators commute when they carry the same Pauli string $P$, including $P=I$, because the position operators commute. Their ideal phase gates can therefore be applied one after another without splitting error. Such identities apply to any finite product of sines and cosines of linear coordinate combinations, but the number of ridge terms can grow with the number of factors.

The elementary compiler input is thus the ridge vector $\bmu\in\bbR^M$, offset $\phi\in\bbR$, dimensionless phase angle $\Lambda\in\bbR$, and Hermitian Pauli string $P$ defining $U_P(\Lambda;\bmu,\phi)$ in Eq.~\eqref{eq:pauli-ridge-target}.  

Throughout, $\norm{\cdot}$ denotes the operator norm.  Our synthesis guarantee is a global bound in this norm under exact signal access. Exact signal access means that the signal operation defined below is implemented without Fock truncation or pulse error.  A physical implementation additionally declares modewise Fock cutoffs and accounts separately for representation error, signal error, and leakage.  Section~\ref{sec:resources} specifies these working-space conditions and their error accounting.

\paragraph{Scope.}
The arrow in Eq.~\eqref{eq:molecular-vibronic} includes modeling approximations and is not an identity on the original electronic Hilbert space. Momentum-dependent derivative couplings require additional signal operations and a separate analysis. Any discrepancy between the retained potential and the original molecular potential belongs to the representation error (Appendix~\ref{app:working-space}).

\subsection{Native signal: conditional momentum kicks}
\label{sec:native-signal}

The ridge phase is built from conditional displacements. A conditional displacement moves the oscillator in opposite directions for the two states of an auxiliary qubit. We call this qubit the synthesis ancilla $a$. Rotations of this ancilla control interference between the displaced paths. The native displacement primitive is

\begin{equation}
\CDx{\bmu/2}
=\exp\!\left[\frac{\ii}{2}
 (\bmu\cdot\hat{\bm Q})\sigma_x^{(a)}\right].
\label{eq:tgifs-elementary-qsp}
\end{equation}

We call this operation a \emph{balanced query}. It applies $\exp[\pm\ii\bmu\cdot\hat{\bm Q}/2]$ to the $\sigma_x^{(a)}$ eigenstates $\ket{\pm}_a$. These eigenstates receive opposite momentum kicks, with relative signal $\exp[\ii\bmu\cdot\hat{\bm Q}]$. Processing rotations act on the ancilla between the queries. This scalar sequence acquires its electronic-state dependence through the boundary Pauli controls of Sec.~\ref{sec:controlled-lift}. The primitive is realized approximately by a spin-dependent force~\cite{McGarry2026AnharmonicDynamics,Eickbusch2022}, and Sec.~\ref{sec:resources} converts its kick action into a fixed-drive CD clock with separate rotation and electronic-control costs.

\subsection{Nonlinear phases as displacement series}
\label{sec:displacement-bridge}

The native CD is linear in the coordinate, but interference between displaced paths can produce a nonlinear phase. To make this connection explicit, define the multimode displacement
\begin{equation}
\mathcal D(\bm\alpha)
=\prod_{j=1}^{M}\exp(\alpha_j\hat a_j^\dagger-\alpha_j^*\hat a_j).
\label{eq:displacement-definition}
\end{equation}
Thus $\exp(\ii n\bm\nu\cdot\bQ)=\mathcal D(\ii n\bm\nu/\sqrt2)$ is a momentum displacement. Consider one real trigonometric group $v_g(\theta)$ with integer harmonics of the collective coordinate $\theta=\bm\nu_g\cdot\bQ$. Its nonlinear phase admits the exact series
\begin{equation}
\ee^{-\ii t v_g(\bm\nu_g\cdot\bQ)}
=\sum_{n\in\mathbb Z}C_{g,n}(t)\,
\mathcal D\!\left(\frac{\ii n\bm\nu_g}{\sqrt2}\right).
\label{eq:nonlinear-displacement-series}
\end{equation}
The coefficients $C_{g,n}(t)$ are the Fourier coefficients of the exponentiated group. The finite trigonometric profile ensures absolute uniform convergence, so the identity also holds in operator norm under exact signal access. Different collective directions or incommensurate frequencies need separate groups for the one-signal construction.

Equation~\eqref{eq:nonlinear-displacement-series} expresses the target as a coherent weighted sum of displacements. It does not prescribe applying those terms sequentially. A balanced CD supplies half a displacement order on each ancilla branch, and intervening rotations determine how the paths interfere. Section~\ref{sec:common-architecture} constructs the coefficients by GJA expansion and turns a finite approximation to this series into a unitary native circuit.

\section{Serial GJA compilation with native QSP}
\label{sec:common-architecture}

We compile the molecular phase calls by grouping compatible harmonics, constructing each group's displacement response, and recovering its CD--rotation circuit. For one electronic Pauli string, write the supplied coordinate function as $f_\ell(\bQ)=c_\ell+\sum_gv_{\ell g}(\bm\nu_{\ell g}\cdot\bQ)$, where each $v_{\ell g}$ has integer harmonics of one base signal. Since the coordinates commute,
\begin{equation}
\begin{split}
\ee^{-\ii\tau P_\ell f_\ell(\bQ)}
={}&\ee^{-\ii\tau c_\ell P_\ell}
\prod_g\ee^{-\ii\tau P_\ell v_{\ell g}(\bm\nu_{\ell g}\cdot\bQ)}.
\end{split}
\label{eq:serial-gja-groups}
\end{equation}
The constant term is an electronic rotation. Each remaining factor is synthesized as one GJA word, and these words are applied in sequence. This \emph{serial GJA} construction introduces no splitting error between compatible coordinate groups, but each finite circuit still contributes its own synthesis error. Noncommuting Pauli sectors and kinetic evolution remain part of the outer algorithm.

The classical compiler evaluates the GJA coefficients, truncates and completes the response, and factors the completed matrix into native gates. The device executes the recovered sequence with electronic controls at its boundaries. Appendix~\ref{app:angle-recipe} gives angle recovery and reconstruction checks.

\subsection{GJA coefficients for compatible harmonic groups}
\label{sec:ja-backend}
\label{sec:generalized-ja}

The first step expresses a nonlinear phase using powers of the available signal. For one ridge, set $\hat\theta=\bmu\cdot\hat{\bm Q}+\phi$ and define

\begin{equation}
W_{\bmu,\phi}
=\exp[\ii(\bmu\cdot\hat{\bm Q}+\phi)].
\label{eq:signal}
\end{equation}

The Jacobi--Anger identities express a nonlinear sinusoidal phase as a series of momentum displacements $W_{\bmu,\phi}^n$, each with a known scalar phase:

\begin{align}
\ee^{-\ii\Lambda\cos\theta}
&=\sum_{n=-\infty}^{\infty}(-\ii)^nJ_n(\Lambda)\ee^{\ii n\theta},\nonumber\\
\ee^{-\ii\Lambda\sin\theta}
&=\sum_{n=-\infty}^{\infty}(-1)^nJ_n(\Lambda)\ee^{\ii n\theta},
\label{eq:ja}
\end{align}

Here $J_n$ is the Bessel function of the first kind~\cite{DLMF,WatsonBessel}. Substitution of the self-adjoint coordinate $\hat\theta$ gives the displacement series for a single-tone gate, the JA special case of the grouped construction below. The position quadratures strongly commute, so the uniformly convergent scalar series transfers to the operator norm.

With $z=\ee^{\ii\theta}$, truncation at degree $K$ gives a Laurent polynomial, a finite sum of positive and negative powers of $z$:

\begin{equation}
S^{c}_{K,\Lambda}(z):=\sum_{n=-K}^{K}(-\ii)^nJ_n(\Lambda)z^n.
\label{eq:ja-raw-cosine-response}
\end{equation}

The displacement powers are

\begin{equation}
W_{\bmu,\phi}^{n}
=\ee^{\ii n\phi}\prod_{j=1}^{M}\ee^{\ii n\mu_j\hat Q_j}.
\label{eq:ridge-factorization}
\end{equation}

Thus the degree remains that of one scalar series. The vector $\bmu$ specifies the physical displacements.

For a compatible harmonic group, the same JA identities supply the combined coefficients before truncation and completion, allowing all harmonics to share one circuit.

For example, consider $h(\theta)=a\sin\theta+b\sin2\theta$. The Fourier coefficients of its phase response $\ee^{-\ii h(\theta)}$ are $C_n=\sum_{\ell\in\mathbb Z}J_{n-2\ell}(-a)J_\ell(-b)$. Both harmonics use integer powers of the same signal $\ee^{\ii\theta}$. Their combined response can therefore be truncated and completed as one word. This common signal is what permits grouping.

To distinguish this case from arbitrary frequencies, consider the general scalar profile
\begin{equation}
g(x)=\sum_{r=1}^{R}a_r\sin(\omega_r x+\phi_r),
\qquad \omega_r\in\bbR,
\label{eq:multitone-ridge-profile}
\end{equation}
Here $a_r$ is a phase amplitude, $\omega_r$ is a spatial frequency, and $\phi_r$ is an offset. All are real and the list is finite. The corresponding oscillator phase is $f(\bQ)=g(\hat x)$, with $\hat x=\bm\xi\cdot\bQ$. A tone denotes one spatial frequency in $x$. It is not a separate oscillator mode. Equation~\eqref{eq:multitone-ridge-profile} describes the general expansion. The one-signal grouped circuit below requires the additional condition $\omega_r=m_r\beta$ with integer harmonic orders $m_r$. Thus the ratios $\omega_r/\beta$ must be integers. The spatial frequencies themselves need not be integers.

The classical starting point for its synthesis is the generalized Jacobi--Anger expansion~\cite{DattoliEtAl1992GeneralizedBessel,DattoliTorreLorenzutta1998,KuklinskiHague2021}.  The ordinary identity for each tone gives absolutely convergent series whose product is
\begin{equation}
\ee^{-\ii g(x)}
=\sum_{\bm k\in\mathbb Z^R}
\left[\prod_{r=1}^{R}J_{k_r}(-a_r)\ee^{\ii k_r\phi_r}\right]
\ee^{\ii(\sum_r k_r\omega_r)x}.
\label{eq:generalized-ja-multitone}
\end{equation}
This identity is uniformly convergent on the real line for finite $R$.  It requires no integer-frequency assumption.  It also includes cosine terms through a $\pi/2$ offset shift.

When the frequencies share a common base, $\omega_r=m_r\beta$, the profile becomes multiharmonic in $\theta=\beta x$.  Constant terms can be implemented separately. Frequency signs can be absorbed into the amplitudes and offsets. We therefore take $\beta>0$ and $m_r\in\mathbb N$.  Define $h(\theta)=\sum_ra_r\sin(m_r\theta+\phi_r)$.  The coefficients of equal powers in Eq.~\eqref{eq:generalized-ja-multitone} combine to give
\begin{align}
G(\ee^{\ii\theta})&:=\ee^{-\ii h(\theta)}
=\sum_{n\in\mathbb Z}C_n\ee^{\ii n\theta},
\label{eq:generalized-ja-response}\\
C_n&=\sum_{\substack{\bm k\in\mathbb Z^R\\\sum_rm_rk_r=n}}
\prod_{r=1}^{R}J_{k_r}(-a_r)\ee^{\ii k_r\phi_r}.
\label{eq:generalized-ja-coefficients}
\end{align}
These generalized Bessel coefficients have one Fourier index. Let $\bm\nu=\beta\bm\xi$ be the base displacement vector. The coefficient $C_n$ weights the momentum order $n\bm\nu$ in the exponential response, whereas the amplitude $a_r$ specifies the coupling itself.

The oscillator--qubit construction follows by substituting the collective coordinate into this scalar response.  Use the native displacement signal
\begin{equation}
W_{\bm\nu}=\exp(\ii\bm\nu\cdot\bQ),\qquad
\exp[-\ii g(\hat x)]=G(W_{\bm\nu}).
\label{eq:generalized-ja-operator}
\end{equation}
Joint functional calculus transfers the uniform approximation to operator norm.  The coefficients supply one response for the entire commensurate group, rather than a separate circuit for each harmonic.

For incommensurate tones, Eq.~\eqref{eq:generalized-ja-multitone} remains valid. However, the corresponding exponentials cannot in general be written as integer powers of one periodic signal.  One-coordinate dependence alone therefore does not give the one-signal factorization above.  A finite-window reapproximation is a separate option.

For a finite response we retain $|n|\le K$ and bound the omitted absolute tail by $\tau_K\ge\sum_{|n|>K}|C_n|$. Section~\ref{sec:guarantees-costs} relates this tail to precision and cost.

\subsection{QSP implementation and unitary completion}
\label{sec:qsp-preliminaries}
\label{sec:response-to-circuit}

\paragraph{QSP as an interference sequence.}
QSP alternates a signal-dependent operation with chosen qubit rotations. Here the signal is the balanced conditional displacement, including a coordinate offset:
\begin{equation}
\begin{split}
U_{\rm circ}&=R_N S_{\phi}R_{N-1}\cdots S_{\phi}R_0,\\
S_{\phi}&=\CDx{\bmu/2}R_x(-\phi).
\end{split}
\label{eq:qsp-execution-template}
\end{equation}
In this ordered gate sequence (a \emph{word}), the rightmost operation acts first. The fixed rotation $R_x(-\phi)$ supplies the coordinate offset and is called an offset carrier. Each $R_j$ rotates the ancilla. Each $S_\phi$ supplies the phases $\pm(\bmu\cdot\hat{\bm Q}+\phi)/2$ in its $X$ basis. The target response determines the rotations and the query count. Fixed input and output basis changes are included in $R_0$ and $R_N$.

For a scalar value $\theta$ of $\bmu\cdot\hat{\bm Q}+\phi$, write $z=\ee^{\ii\theta}$. With $N=2K$ balanced queries, the matrix entries contain integer powers between $z^{-K}$ and $z^K$. The Laurent degree is therefore at most $K$. Their coefficients describe the resulting response, not the rotation angles themselves.

\paragraph{Why unitary completion is needed.}
An approximation to the desired response specifies only one part of the circuit. For an input ancilla in $\ket0$, let $F(z)$ be the amplitude returning to $\ket0$ (the selected branch), and $H(z)$ the amplitude reaching $\ket1$. Then unitarity requires $|F(z)|^2+|H(z)|^2=1$ on the entire signal circle. A truncated Fourier response can violate $|F|\le1$. Such a response must be repaired before completion. The completion step constructs a compatible complementary polynomial. Factorization of the completed matrix gives the qubit rotations. The construction below enforces contractivity before recovering the circuit.

To reuse the ancilla coherently, we approximate a full gate. In our construction, the target phases on the two ancilla states have opposite signs. The paired target is
\begin{equation}
\widehat V_a:=
\begin{pmatrix}V&0\\0&V^\dagger\end{pmatrix}_a.
\label{eq:paired-completion}
\end{equation}
Here $V$ is the desired oscillator unitary. The states $\ket0$ and $\ket1$ apply $V$ and $V^\dagger$, respectively. We call $\|U-\widehat V_a\|$ the \emph{paired-unitary error}. A comparison of only $\bra0U\ket0$ with $V$ can hide the complementary amplitude. Our construction controls the paired operation without measurement or reset.

A contractive response $F$ and a complementary Laurent polynomial $H$ give a matrix that is unitary at every signal value, or paraunitary. Substitution of the oscillator signal gives
\begin{equation}
U_{\rm circ}=
\begin{pmatrix}
F(W)&-H(W)^{\dagger}\\
H(W)&F(W)^{\dagger}
\end{pmatrix},
\label{eq:ja-operator-block-form}
\end{equation}

For an ancilla initialized in $\ket0$, this matrix acts as
\begin{equation}
\begin{split}
U_{\rm circ}(\ket0\otimes\ket\psi)
={}&\ket0\otimes F(W)\ket\psi\\
&+\ket1\otimes H(W)\ket\psi.
\end{split}
\label{eq:qsp-state-action}
\end{equation}
The desired output has $F\approx V$ and $H\approx0$. The paired-unitary error bound controls both ancilla return and the action on the other input branch.

For a grouped response $G$, contractivity can be enforced by rescaling. For a certified bound $\tau_K\ge\sum_{|n|>K}|C_n|$, the truncation
\begin{equation}
F_K(z)=\frac{\sum_{|n|\le K}C_nz^n}{1+\tau_K}
\label{eq:generalized-ja-contractive}
\end{equation}
obeys $|F_K|\le1$ and $|F_K-G|\le2\tau_K$ on the whole signal circle.  One paraunitary completion and inverse factorization then give an OQ--GQSP word with $2K$ balanced queries~\cite{Haah2019QSP,MotlaghWiebe2024,HongOQGQSP2026}.  The coefficients can be evaluated by Bessel sums, FFTs, or recurrence methods as described in Appendix~\ref{app:angle-recipe}.

For the single-tone JA response, a repair step adjusts the truncated coefficients. In addition to contractivity, the repair enforces the endpoint conditions for a phase-only core when the target is a pure sine or cosine. In such a core, all variable rotations are about $Z$. The maximum harmonic does not increase (Appendix~\ref{smsec:ja-details}).

Phase recovery factors the completed matrix into $2K$ balanced queries and intervening rotations [Fig.~\ref{fig:tgifs-ja-nlft-dictionary}(a)]. It uses inverse nonlinear Fourier transform (NLFT) synthesis under the GQSP--NLFT correspondence~\cite{LaneveNLFT2025}. The balanced-query exponents $\pm1/2$ generate the powers $-K,\ldots,K$. Responses other than pure sines and cosines may require arbitrary $\mathrm{SU}(2)$ rotations. Appendix~\ref{app:angle-recipe} describes the reference recovery. It also describes the alternative route through the inverse nonlinear fast Fourier transform (INLFFT) with its completion and stability conditions~\cite{NiInverseNLFFT2025}.

Completion also connects the selected-response error to the paired-unitary error.  Consider a unitary target $V$ and a completed word of Eq.~\eqref{eq:ja-operator-block-form}. Suppose that all blocks are functions of the same unitary signal $W$. Then a selected-response error $\eta$ gives a paired-unitary error of at most $\sqrt{2\eta}$ [Lemma~\ref{cor:ja-full-joint-completion}].  The choice $\eta\le\eps^2/2$ therefore gives a paired-unitary error of at most $\eps$.

At fixed amplitudes and harmonic indices, contractive rescaling and completion give a paired-unitary error of at most $2\sqrt{\tau_K}$. The coefficient tail still decays faster than any geometric rate (Sec.~\ref{sec:guarantees-costs}), and the square root changes only the rate constant.

\paragraph{Completion freedom.}
The same selected response $F$ can admit different complementary polynomials $H$. Choices of reciprocal roots in spectral completion preserve $|H|^2=1-|F|^2$ and hence the ideal ancilla-return probability of a single group, but change the rotations and intermediate circuit dynamics. Section~\ref{sec:composition} uses this freedom to reduce coherent error between serially composed words.

\subsection{Pauli-conditioned multimode lift}
\label{sec:controlled-lift}

We now add the electronic Pauli string to the completed scalar circuit.  The two ancilla branches already apply opposite phases.  An ancilla flip when $P$ has eigenvalue $-1$ exchanges those branches and supplies the required sign.  Electronic control can therefore act only at the boundaries. The scalar synthesis is unchanged.

\label{lem:boundary-pauli-lift}
To write this boundary construction explicitly, let $g(\bQ)$ be self-adjoint and let $P=P^\dagger$, $P^2=I$ act on a separate electronic register. Suppose the scalar circuit $U_g$ approximates
\begin{equation}
D_g=\exp[-\ii\sigma_z^{(a)}g(\bQ)]
=\operatorname{diag}(\ee^{-\ii g(\bQ)},\ee^{\ii g(\bQ)})_a.
\label{eq:boundary-scalar-target}
\end{equation}
With identities on unused registers implicit, define
\begin{equation}
\Pi_\pm=\frac{I\pm P}{2},\qquad
B_P=I_a\otimes\Pi_++\sigma_x^{(a)}\otimes\Pi_-.
\label{eq:boundary-pauli-involution}
\end{equation}
Then $B_P=B_P^\dagger=B_P^{-1}$ and
\begin{align}
B_PD_gB_P&=\exp[-\ii\sigma_z^{(a)}P g(\bQ)],\nonumber\\
\norm{B_PU_gB_P-B_PD_gB_P}&=\norm{U_g-D_g}.
\label{eq:boundary-pauli-error}
\end{align}
The equality concerns the complete unitary on arbitrary joint inputs, not only a selected or postselected block.

For a nonidentity Pauli string choose an electronic Clifford $C_P$ that maps $P$ onto a pivot qubit $p$, with $C_PPC_P^\dagger=\sigma_z^{(p)}$.  Then $B_P=C_P^\dagger\operatorname{CNOT}_{p\to a}C_P$, and the lifted word is
\begin{equation}
C_P^\dagger\operatorname{CNOT}_{p\to a}\,U_g\,
\operatorname{CNOT}_{p\to a}C_P.
\label{eq:boundary-pauli-circuit}
\end{equation}
In circuit time order, $C_P$ acts first and $C_P^\dagger$ last.  The intermediate Clifford factors cancel because $U_g$ acts trivially on the electronic register.  The lift uses two logical pivot--ancilla CNOTs and one Clifford compute/uncompute pair. This overhead is independent of the scalar query count.  $P=I$ needs no boundary gates.  With the ancilla initially in $\ket0$, the ideal endpoint implements $\exp[-\ii P g(\bQ)]$.

Appendix~\ref{app:boundary-controls} gives the circuit diagrams, proof, and boundary sharing across consecutive words. The same boundary construction applies to the completed grouped word with $g(\bQ)=h(\bm\nu\cdot\bQ)$.  This gives $\exp[-\ii\sigma_z^{(a)}P h(\bm\nu\cdot\bQ)]$ with the same paired-unitary error and the same $2K$ balanced queries, for mixtures of even and odd harmonics alike.  The group must carry one common $P$.  Appendix~\ref{app:joint-signal-access} contrasts this lift with the direct joint-signal substitution used in the formal single-tone bounds. Neither route merges incommensurate frequencies or distinct Pauli sectors into one scalar signal.

\subsection{Coherent composition of compiled words}
\label{sec:composition}

The complete paired-unitary guarantee allows the compiled phase gates to share one synthesis ancilla without measurement or reset after each call. Their coherent product forms the phase-gate sequence within an outer evolution step. For noncommuting targets, composition keeps the order $\widehat V_B\cdots\widehat V_1$, and its lower block need not equal $(V_B\cdots V_1)^\dagger$. Section~\ref{sec:resources} allocates the errors and costs.

Equation~\eqref{eq:serial-gja-groups} supplies the compatible groups for this serial composition. Each completed word carries its own degree, error allowance, and displacement cost, and the paired-unitary errors add by unitary telescoping. This remains a sequence of one-variable syntheses, and a global multivariable response alone does not supply its completion or native factorization (Appendix~\ref{app:multidirection-response}).

Completion can also change the coherent error of a serial product at fixed selected responses and degrees. For two completed words of Eq.~\eqref{eq:ja-operator-block-form},
\begin{equation}
\bra0 U_2U_1\ket0=F_2F_1-H_2^\dagger H_1.
\label{eq:serial-completion-interference}
\end{equation}
The second term describes amplitude that leaves the selected ancilla branch and returns in the next group. Its interference depends on the complementary phases, even when the ideal coordinate phases commute. Each group's paired-unitary error and $2K_g$ query count remain unchanged, as does their telescoping bound. A simple rule uses this freedom without a search. For a word applied in successive outer steps, conjugation of every second application by $\sigma_z^{(a)}$ replaces $H$ by $-H$. At each signal value, the amplitude that leaks and returns over the two applications then changes from $2|H||\operatorname{Re}F|$ to $2|H||\operatorname{Im}F|$. The alternation therefore suppresses the leading leakage when the phase per step is small, and the conjugation is absorbed into existing rotations without adding CD pulses. Intervening groups and free evolution make this cancellation approximate. Section~\ref{sec:short-window-dynamics} shows that it admits a lower degree allocation, an indirect route to lower cost rather than a reduction in queries at fixed degree.

For one step of the outer evolution, the coordinate-dependent propagator has the form
\begin{equation}
\exp\!\left[-\ii\Delta t
\sum_{\ell=1}^{L}P_\ell\otimes f_\ell(\hat{\bm Q})\right].
\label{eq:coordinate-step}
\end{equation}
Mutually commuting Pauli strings may be diagonalized together and implemented as bosonic phases conditioned on the corresponding electronic sectors.  Appendix~\ref{app:working-space} separates these phase calls into Gaussian, ridge, and polynomial pieces [Eq.~\eqref{eq:block-phase-model}].

Product formulas use these gates directly. Interaction-picture Dyson and qubitization methods additionally need controlled Hamiltonian-access or block-encoding oracles and their costs~\cite{MaxwellTrotterError2026,CasaresSPRINT2026,LowWiebe2018IP,KieferovaScherer2019,Eklund2026EndToEnd}. A list of phase calls alone does not supply those oracles.

\section{Compilation guarantees and comparison with existing methods}
\label{sec:complexity-results}

The construction separates two decisions: how to represent the physical potential and how to synthesize its phase gate. We first establish precision scaling relative to repeated weak-cell synthesis, then distinguish potential and gate accuracy and examine their allocation across coherent sequences.

\subsection{Precision scaling and resource guarantees}
\label{sec:guarantees-costs}
\label{sec:trig-fourier-scope}

Trigonometric-gate-implemented Fourier synthesis (TGIFS) is a product formula that generates a sinusoidal phase by repeating weak cells built from conditional displacements and small ancilla rotations~\cite{McGarry2026AnharmonicDynamics,Chalermpusitarak2025TGIFS,Rainaldi2026TrigBenchmark}. For a zero-offset cosine, let $C=\CDx{\bmu/2}$ and $A_\pm=C^{\pm1}\sigma_z^{(a)}P C^{\mp1}$. Then
\begin{equation}
\begin{split}
A_++A_-&=2\sigma_z^{(a)}P\cos(\bmu\cdot\bQ),\\
\mathcal G_P^{\rm basic}(\vartheta)
&=\ee^{\ii\vartheta A_+}\ee^{\ii\vartheta A_-},\\
\widehat U_{a,P}(\Lambda)
&\approx\left[\mathcal G_P^{\rm basic}\!\left(-\frac{\Lambda}{2r}\right)\right]^r.
\end{split}
\label{eq:tgifs-main-comparison}
\end{equation}
Here $\widehat U_{a,P}=\operatorname{diag}(U_P,U_P^\dagger)_a$ is the paired target for $\phi=0$. The generators generally do not commute, so the cell approximates a small phase and repetition controls the accumulated error. The first-order construction has the sufficient paired-unitary bound $\Lambda^2/(4r)$ [Eq.~\eqref{eq:small-angle-lt-error}]. We also test a symmetric second-order $A$--$B$--$A$ extension. Its exact cell and query cost are given in Appendix~\ref{app:tgifs-second-order}, and the first-order construction and its bound in Appendix~\ref{smsec:tgifs-details}.

JA instead constructs and completes the finite-phase response before recovering the native rotations. This avoids the precision overhead of the tested fixed-order weak-cell repetitions within one phase gate. Figure~\ref{fig:tgifs-ja-nlft-dictionary} compares the circuit structures and convergence. Product-formula assembly of distinct noncommuting Hamiltonian terms is a separate outer approximation.

For single-tone JA, the finite construction uses the same signal principle as the synthesis of $\exp[\ii t\cos H]$ and $\exp[\ii t\sin H]$ from $U=\exp(\ii H)$ in Ref.~\cite{MotlaghWiebe2024}. Here the signal is $W_{\bmu,\phi}$. The Jacobi--Anger expansion supplies its response coefficients.  Reference~\cite{Rainaldi2026TrigBenchmark} implements the cosine gate by a product-formula circuit with postselection and uses the expansion to analyze the resulting states.  Our completed word keeps the ancilla coherent. The boundary construction in Sec.~\ref{sec:controlled-lift} extends this coherent operation to electronic Pauli conditioning.

We measure circuit cost by the CD time $T_{\rm CD}$ in units of one base balanced query, $\tau_{\rm base}$, so that $\CDx{s\bmu/2}$ costs $|s|\tau_{\rm base}$ (Sec.~\ref{sec:resources}). After boundary cancellations, the first-order and symmetric second-order TGIFS words cost $4r\tau_{\rm base}$ and $(4r+2)\tau_{\rm base}$.

\begin{figure*}[!t]
\centering
\includegraphics[width=0.80\textwidth]{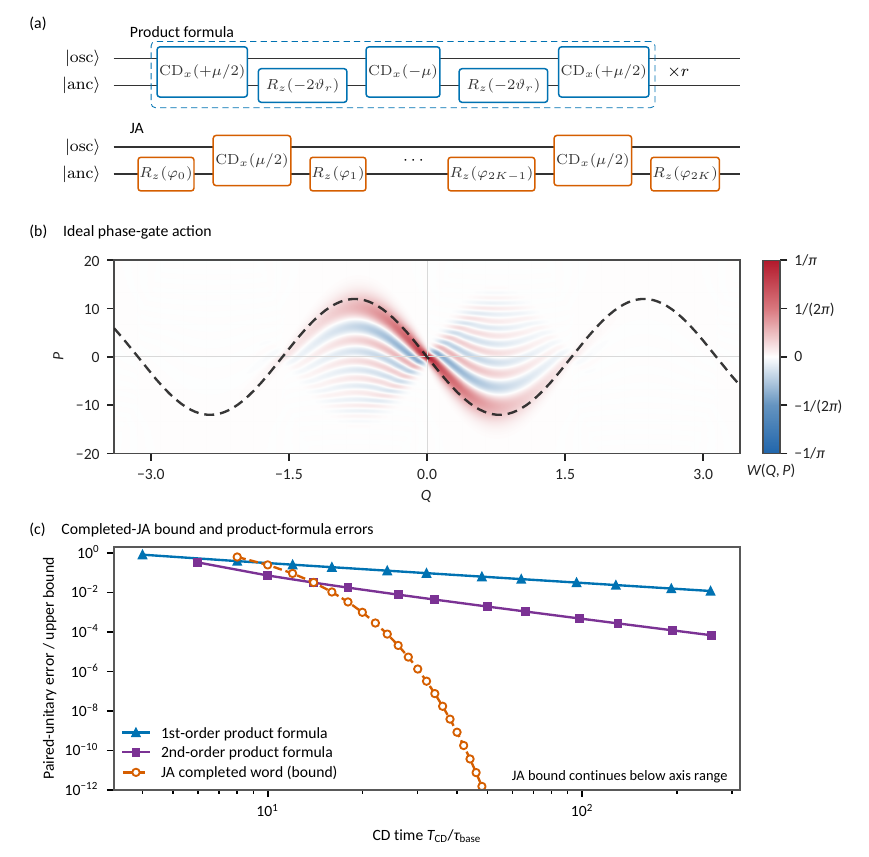}
\caption{Single-tone circuit, ideal-state illustration, and error comparison. Panels (b) and (c) use separate parameter choices. (a) A product formula (TGIFS~\cite{McGarry2026AnharmonicDynamics,Chalermpusitarak2025TGIFS,Rainaldi2026TrigBenchmark}) repeats a weak cell within one phase gate. JA uses $2K$ balanced queries and compiled $Z$ rotations for the paired ancilla target $\exp[-\ii\Lambda\sigma_z^{(a)}\cos(\mu Q)]$, with electronic $P=I$. (b) Ideal vacuum-output Wigner function for the opposite ancilla branch, $\exp[+\ii\Lambda\cos(\mu Q)]$, at $(\Lambda,\mu)=(6,2)$. The vertical axis $P$ is the momentum conjugate to $Q$. It is not the electronic Pauli string. The dashed line is $P=-\Lambda\mu\sin(\mu Q)$. (c) Sampled paired-unitary errors at $\Lambda=2$ for first-order (Lie--Trotter) and second-order (Strang) product formulas. The JA curve is a theoretical upper bound for the completed word. It is not a recovered-circuit error. Appendix~\ref{app:fig2-bound} gives its evaluation. Recovered circuits are compared in Fig.~\ref{fig:compiler-benchmarks}.}
\label{fig:tgifs-ja-nlft-dictionary}
\end{figure*}

\paragraph{GJA tail bounds and resource guarantees.}

The convergence is controlled by the factorial decay of the Bessel coefficients. For real $\Lambda$, $|J_n(\Lambda)|\le(|\Lambda|/2)^n/n!$ for $n\ge0$, so the uniform truncation error satisfies
\[
\sup_\theta\bigl|\ee^{-\ii\Lambda\cos\theta}-S^c_{K,\Lambda}(\ee^{\ii\theta})\bigr|
\le 2\sum_{n>K}\frac{(|\Lambda|/2)^n}{n!}.
\]
The same bound holds for the sine response. At fixed nonzero amplitude, it decays superexponentially as $\exp[-K\log K+\calO(K)]$, the same decay used in trigonometric generalized QSP (GQSP) synthesis~\cite{MotlaghWiebe2024}. Larger amplitudes require more displacement orders before this regime is reached (Appendix~\ref{smsec:ja-details}).

Figure~\ref{fig:tgifs-ja-nlft-dictionary}(c) shows a completed-JA error bound from a uniform truncation-tail majorant. Sections~\ref{sec:duration-repricing} and \ref{sec:grouped-benchmarks} compare recovered circuits. Single tones use direct phase-only fitting. Grouped responses use coefficient construction, spectral completion, and matrix factorization. Both routes validate the complete native circuit against its target (Appendix~\ref{app:numerical-validation}).

For $L_{\rm h}=\max_r m_r$ and $S=\sum_r m_r|a_r|>0$, recall the omitted absolute tail $\tau_K=\sum_{|n|>K}|C_n|$ of Sec.~\ref{sec:generalized-ja}. The generalized Bessel coefficients themselves decay superexponentially with Fourier order at fixed amplitudes and harmonic indices. The contour bound in Appendix~\ref{smsec:generalized-ja} gives
\[
|C_n|\le
\left(\frac{\mathrm e S}{2|n|}\right)^{|n|/L_{\rm h}},
\qquad |n|>S/2.
\]
The logarithm of this upper bound is $-(|n|/L_{\rm h})\log|n|+\calO(|n|)$. The bound therefore eventually decreases faster than any fixed geometric rate. The sum of the omitted coefficients gives the corresponding uniform tail bound
\[
\sum_{|n|>K}|C_n|
\le \exp\!\left[-\frac{K}{L_{\rm h}}\log K+\calO(K)\right]
\qquad(K\to\infty).
\]
Thus the coefficient decay also controls the uniform truncation error. A larger $L_{\rm h}$ weakens this sufficient rate. Equation~\eqref{eq:generalized-ja-superexponential-tail} gives an explicit finite-degree majorant.

A grouped word also keeps the cancellation between harmonics that an absolute-amplitude estimate discards. For the combined profile $h$, define
\begin{equation}
B_h(s)=\sup_{\theta\in\bbR}|\operatorname{Im}h(\theta+\ii s)|,
\qquad s>0.
\label{eq:generalized-ja-phase-growth}
\end{equation}
A shift of the coefficient contour gives the sharper bound
\begin{equation}
\tau_K\le\frac{2\exp[B_h(s)-(K+1)s]}{1-\ee^{-s}}.
\label{eq:generalized-ja-phase-aware-tail}
\end{equation}
Because $B_h(s)\le\sum_r|a_r|\sinh(m_rs)$, this bound is no larger than the absolute-amplitude contour bound at the same $s$. It retains the offsets and interference among harmonics before taking a supremum. At small $s$, $B_h(s)=s\|h'\|_\infty+\calO_h(s^3)$. Cancellation in the combined phase gradient can reduce the sufficient cutoff and hence the CD action. Appendix~\ref{app:action-optimality} gives the proof and a two-harmonic example with an explicit finite-accuracy cutoff comparison. The reduction depends on the supplied profile, and a grouped word does not guarantee a strict saving for every target.

For $S>0$ and paired-unitary error tolerance $0<\eps<1/2$, a sufficient degree is
\begin{equation}
\begin{gathered}
K=\calO\!\left(S+
\frac{L_{\rm h}T_{\rm h}}{\log(\mathrm e+L_{\rm h}T_{\rm h}/S)}\right),\\
T_{\rm h}=\log(L_{\rm h}/\eps).
\end{gathered}
\label{eq:generalized-ja-degree}
\end{equation}
Appendix~\ref{smsec:generalized-ja} derives the tail bound and paired-unitary guarantee. The theorems below state the query and displacement-action bounds in their respective circuit models.

The grouped construction truncates and completes the combined response once. Serial synthesis shares the tolerance among separately completed words. Under the fixed-drive CD clock, a serial harmonic-$m_r$ word costs $2|m_r|K_r\tau_{\rm CD}(\bm\nu/2)$, whereas the grouped word costs $2K\tau_{\rm CD}(\bm\nu/2)$. Cancellation within the combined coefficients can lower its bandwidth. The actual comparisons use recovered words.

A generic completion uses $\mathrm{SU}(2)$ rotations. Phase-only alternatives can matter for the physical cost. Separate single-tone compilation provides one phase-only option. Appendix~\ref{app:processor-conventions} gives the symmetry restriction and the separate pulse costs.

Theorem~\ref{thm:jaoqgqsp-trig} gives coherent constructions and a matching single-tone query lower bound. Theorem~\ref{prop:action-optimality} treats fixed-profile query and action scaling. Both assume exact signals and uniform approximation in their stated circuit models, and finite-cutoff and implementation errors are accounted for separately. For the lower bounds, a \emph{uniform signal processor} alternates controlled scalar signal calls, or their inverses, with signal-independent controls on finite ancillary registers, and a power $W^m$ counts as $|m|$ calls.

\begin{theorem}[Coherent sinusoidal compilation]
\label{thm:jaoqgqsp-trig}\label{prop:ja-single-ridge-lower}
Let the position quadratures strongly commute, let $P$ be a Hermitian electronic Pauli string, and let $0<\eps<1/2$.

(i) \emph{Grouped upper bound.} Let $h(\theta)=\sum_ra_r\sin(m_r\theta+\phi_r)$ with $m_r\in\mathbb N$, $L_{\rm h}=\max_rm_r$, and $S=\sum_rm_r|a_r|>0$. One word with $2K$ balanced queries, $K$ as in Eq.~\eqref{eq:generalized-ja-degree}, and boundary Pauli control approximates $\exp[-\ii\sigma_z^{(a)}P\,h(\bm\nu\cdot\bQ)]$ within $\eps$ in operator norm.

(ii) \emph{Single-tone upper bound.} For $U_P(\Lambda;\bmu,\phi)$ with $\Lambda\ne0$, a word with a phase-only core and fixed boundary basis changes approximates $\operatorname{diag}(U_P,U_P^\dagger)_a$ within $\eps$ in operator norm, using
\begin{equation}
N_W=2K=\calO\!\left(|\Lambda|+
\frac{\log(1/\eps)}{\log(\mathrm e+\log(1/\eps)/|\Lambda|)}\right)
\label{eq:jaoqgqsp-query}
\end{equation}
balanced queries.

(iii) \emph{Matching lower bound.} A uniform signal processor whose selected response approximates $e^{-\ii\Lambda\cos\theta}$ within $\eps$ for all real $\theta$ uses $Q$ calls with
\begin{equation}
Q+1=\Omega\!\left(|\Lambda|+
\frac{\log(1/\eps)}{\log(\mathrm e+\log(1/\eps)/|\Lambda|)}\right).
\label{eq:ja-single-ridge-lower}
\end{equation}
For $P\ne I$ and $\bmu\ne\bm0$, the same bound holds for joint-signal processors whose other gates act only on signal-processing ancillas.
\end{theorem}

Appendices~\ref{smsec:ja-details} and \ref{app:proofs} prove the three parts, with the joint-signal convention of parts (ii) and (iii) in Appendix~\ref{app:joint-signal-access}. The query count scales as $\calO(|\Lambda|)$ at large amplitude and as $\calO[\log(1/\eps)/\log\log(1/\eps)]$ at high precision, and the $+1$ in part (iii) accounts for integer queries near the identity. A paired approximation fixes its selected block, so part (iii) also bounds every full paired approximation. The boundary lift preserves the upper bounds but does not extend the lower bound to arbitrary electronically controlled architectures. All tones of a group carry one common $P$, and the zero profile needs no queries. The joint optimal dependence on amplitudes, precision, and growing harmonic content remains open.

\Needspace{8\baselineskip}
\begin{theorem}[Query and action requirements of a fixed profile]
\label{prop:grouped-amplitude-optimality}\label{prop:action-optimality}\label{cor:action-l2}
Fix a real nonconstant trigonometric polynomial $h$ with integer harmonics, $\bm\nu\ne\bm0$, $\lambda>0$, and $0<\eps<1/2$, and let $C_n$ be the Fourier coefficients of $e^{-\ii\lambda h}$.

(i) \emph{Query scaling.} Consider words of balanced queries $D(\theta)=\operatorname{diag}(e^{\ii\theta/2},e^{-\ii\theta/2})$ and their inverses, with finite ancillary registers and signal-independent coherent controls. A word whose selected response approximates $e^{-\ii\lambda h(\theta)}$ within $\eps$ for all real $\theta$ uses $N$ queries with
\begin{equation}
N+2\ge\lambda\Delta_h/\pi,\qquad \Delta_h=\max h-\min h,
\label{eq:grouped-amplitude-lower}
\end{equation}
and the grouped construction attains the minimum paired-unitary query count up to a factor that depends on $h$ and $\eps$,
\begin{equation}
N_{\min}(\lambda;h,\eps)=\Theta_{h,\eps}(\lambda)\quad(\lambda\to\infty).
\label{eq:grouped-amplitude-optimal}
\end{equation}

(ii) \emph{Displacement action.} Consider single-ancilla words with signal-independent $\mathrm{SU}(2)$ rotations and real kicks $e^{\ii\alpha_j\sigma_z^{(a)}\theta}$, with action $A=|\bm\nu|\sum_j|\alpha_j|$. Among words that approximate $\operatorname{diag}(e^{-\ii\lambda h},e^{\ii\lambda h})$ within $\eps$ for all real $\theta$,
\begin{equation}
A_{\min}=\lambda|\bm\nu|\,\|h'\|_\infty+o(\lambda)
\quad(\lambda\to\infty),
\label{eq:action-leading-main}
\end{equation}
and every such word has $A\ge|\bm\nu|K_2(\eps)$ at every amplitude, where
\begin{equation}
K_2(\eps)=\min\left\{m\in\mathbb Z_{\ge0}:\sum_{|n|>m}|C_n|^2\le\eps^2\right\}.
\label{eq:action-l2-order}
\end{equation}
\end{theorem}

Part (i) fixes the base kick and establishes linear query scaling at fixed profile and accuracy. Part (ii) allows unequal real kick strengths along the same coordinate and determines the leading action coefficient. The quantity $\lambda|\bm\nu|\|h'\|_\infty$ is the largest momentum impulse of the ideal phase. Along the normalized ridge coordinate, $\exp[-\ii\lambda h(\bm\nu\cdot\bQ)]$ shifts momentum by $-\lambda|\bm\nu|h'(\bm\nu\cdot\bQ)$. Reference~\cite{Yang2026RabiPolyPhase} bounds the total Rabi time of same-quadrature words by the derivative of their postselected block, for polynomial phase gates on low-energy inputs. Part (ii) treats sinusoidal profiles under uniform approximation of the paired unitary and fixes the leading coefficient. The grouped construction attains this leading coefficient through the phase-aware tail bound of Appendix~\ref{app:action-optimality}. That bound sharpens the general sufficient bound of Eq.~\eqref{eq:generalized-ja-degree}. The action statement is restricted to the specified single-ancilla same-coordinate model. Conjugate-quadrature displacements, adaptive measurements, and restricted-input approximation are outside its scope. A Gaussian squeezing that rescales $\bm\nu\cdot\bQ$ also rescales the displacement action, so action comparisons fix the coordinate frame.

At finite amplitude, $K_2(\eps)$ gives different information. It is the smallest momentum-order cutoff that retains all but $\eps^2$ of the diffraction weight of the target. It gives the necessary bound $A\ge|\bm\nu|K_2$, but does not specify a circuit attaining that action. The GJA coefficients keep the factorial-type decay of the single-tone Jacobi--Anger series, slowed only by the largest harmonic index $L_{\rm h}$. The degree at which our tail bounds guarantee paired-unitary error $\eps$ and the necessary degree $K_2(\eps)$ are therefore set by the same tail. For the single-tone target at $\Lambda=2$ and the eight-harmonic target of Eq.~\eqref{eq:dense-eight-profile} at $s=1$, their ratio stays between $1.6$ and $2.2$ for $10^{-10}\le\eps\le10^{-2}$. The recovered degree $K$, the necessary support $K_2$, and the asymptotic impulse scale need not coincide, and whether circuits can approach the action $|\bm\nu|K_2(\eps)$ at finite amplitude remains open.

\subsection{Potential accuracy and gate accuracy}
\label{sec:hamiltonian-first}

Earlier OQ--GQSP synthesis approximates the desired unitary phase directly by a finite Laurent response~\cite{HongOQGQSP2026}. For one signal coordinate, the two routes can be written as
\begin{equation}
\begin{aligned}
\text{unitary-first:}\quad
&u_t=\ee^{-\ii tV}\longrightarrow P_K(z),\\
\text{Hamiltonian-first:}\quad
&V\longrightarrow V_R,\\
&u_{R,t}=\ee^{-\ii tV_R}\longrightarrow S_K(z).
\end{aligned}
\label{eq:approximation-routes}
\end{equation}
Both finite responses must satisfy whole-circle contractivity before QSP completion and factorization. In our route, GJA supplies the coefficients of $u_{R,t}$ from the real trigonometric representation $V_R$. The distinction concerns which function is approximated first, rather than a different QSP backend. The two routes still give different approximation families, because an exponential of a finite real trigonometric model generally has infinitely many Fourier coefficients.

For a real potential $V$ and its real representation $V_R$ on a declared coordinate domain $\mathcal I_Q$, define $\delta_V=\|V-V_R\|_{\mathcal I_Q}$. Before finite-response synthesis,
\begin{equation}
\begin{split}
|u_{R,t}|&=1,\\
\|u_t-u_{R,t}\|_{\mathcal I_Q}&\le |t|\delta_V.
\end{split}
\label{eq:representation-phase-bound}
\end{equation}
Representation error therefore changes phase without introducing amplitude loss. A real Fourier extension preserves this unit modulus outside the physical interval too. A direct complex-response fit can have both phase and amplitude error, and accuracy inside the interval alone does not ensure whole-circle contractivity. Phase-aware unitary extensions can also preserve unit modulus before finite-response synthesis. Appendix~\ref{app:working-space} gives the domain and extension conditions.

The subsequent truncation and completion have their own error. For $\|S_K-u_{R,t}\|_{\mathbb T}\le\tau_K<1$, take $F_K=S_K/(1+\tau_K)$ as in Eq.~\eqref{eq:generalized-ja-contractive}. On the signal circle,
\begin{equation}
\begin{split}
|H_K|^2=1-|F_K|^2
&\le\frac{4\tau_K}{(1+\tau_K)^2},\\
\epsilon_{\rm pair,total}
&\le |t|\delta_V+2\sqrt{\tau_K}.
\end{split}
\label{eq:representation-leakage-separation}
\end{equation}
The first line bounds the ancilla leakage probability at a scalar signal value and, under exact signal access, its expectation for any oscillator input. The second line bounds the paired-unitary error on the declared coordinate domain. Thus potential-series approximation and finite-response synthesis have distinct certificates. The response tail still depends on $V_R$, so these are not independent optimization problems.

Figure~\ref{fig:structure-programming}(a) illustrates the two approximation routes in the complex plane. The Hamiltonian-first intermediate response $u_R$ lies on the unit circle because the potential model is real. Its subsequent finite response $F_{\rm HF}$ can have both phase and amplitude error. The diagram shows the illustrative case of zero synthesis phase error, and its arrows indicate approximation stages rather than error-vector lengths.

For a real periodic potential $V$, let $\Pi_j$ retain Fourier modes $-j,\ldots,j$ and consider $V_R=\Pi_R V$. Before contractivity repair, the approximation order gives
\begin{equation}
\begin{split}
S^{\rm HF}_{R,K}&=\Pi_K\!\left[\ee^{-\ii t\Pi_R V}\right],\\
S^{\rm UF}_{K}&=\Pi_K\!\left[\ee^{-\ii tV}\right].
\end{split}
\label{eq:truncation-order-routes}
\end{equation}
Potential truncation and exponentiation generally do not commute. If $\Pi_RV=V$, however, these responses coincide, as the matched-$V_R$ control of Appendix~\ref{app:sinusoidal-route-comparison} confirms under identical truncation and completion. Within each compatible group, our GJA construction combines the harmonic contributions before the final mode cutoff. A truncation of each harmonic's JA response before multiplication would give a different approximation, and the composition of separately completed groups is also a different circuit operation.

A fixed potential model also defines a consistent family $U_R(t)=\exp(-\ii tV_R)$ with $U_R(t+s)=U_R(t)U_R(s)$. The same $\delta_V$ certificate can be reused at each time, with representation error bounded by $|t|\delta_V$ on its declared domain. The degree, completion, and rotation angles must still be synthesized for each requested time and tolerance, and approximate circuits need not obey the group law exactly. The same full-unitary accuracy criterion permits coherent reuse of either synthesis route when its bounds cover the incoming states.

The explicit GJA construction supplies coefficients and tail bounds without optimizing the finite response, but it does not imply a universal query or runtime advantage over direct synthesis.

\subsection{Error allocation in coherent composition}
\label{sec:resources}

For a serial sequence, the potential degrees and the QSP response degrees jointly determine accuracy and cost. Uniform gate-error allocation gives a sufficient bound, but a finite search can instead select these degrees across compatible groups under one total-error criterion.

\paragraph{Accumulated synthesis error.}
For $m$ serial ridges with amplitudes $\Lambda_i$, choose paired-unitary tolerances $\eps_i$ satisfying $\sum_i\eps_i\le\eps_{\rm in}$.  The uniform allocation $\eps_i=\eps_{\rm in}/m$ is sufficient.  Unitary telescoping gives a measurement-free product with error at most $\eps_{\rm in}$ and
\begin{equation}
N_W^{\rm JA}
=\calO\!\left[\sum_{i=1}^{m}\abs{\Lambda_i}
+m\,\operatorname{polylog}(m/\eps_{\rm in})\right].
\label{eq:aggregate-ja-main}
\end{equation}
Equation~\eqref{eq:aggregate-ja-main} counts queries to each ridge's own signal and is an achievable serial upper bound. The conversion to a common base-signal or CD-action currency additionally weights each harmonic by $|m_i|$ and each direction by its displacement cost. Compression or cancellation across ridge directions that do not share one base signal is the multivariate problem discussed in Appendix~\ref{app:multidirection-response}.

This count applies within an outer time step only with compatible gate-error bounds on all states reached during the evolution, which projected commutator bounds alone do not supply (Appendix~\ref{app:working-space}).

\paragraph{Representation and implementation errors.}
The query bound assumes exact signals and a supplied representation. On a common working space, six error sources add coherently. They are the representation, outer-evolution, synthesis, phase-recovery, signal, and boundary-control errors. Leakage and channel noise retain their own metrics. Appendix~\ref{app:working-space} specifies the projected contract and the error allocation.

\paragraph{Joint allocation across compatible groups.}
Figure~\ref{fig:structure-programming}(b) compares the routes for three distinct compatible groups with profiles $V(Q_1)$, $0.75V(Q_2)$, and $0.4V((Q_1+Q_2)/2)$, where $V(q)=(1-\ee^{-q})^2$. We apply the three-group sequence $16$ times at $\Delta t=2$, retaining one ancilla without measurement or reset. Both routes use the same real periodic extension, global contractivity repair, and completion convention. Hamiltonian-first synthesis varies the potential and response degrees jointly, and unitary-first synthesis varies the response degrees of the directly exponentiated extension. Both routes allocate degrees across the groups.

\begin{figure*}[!t]
\centering
\includegraphics[width=0.98\textwidth]{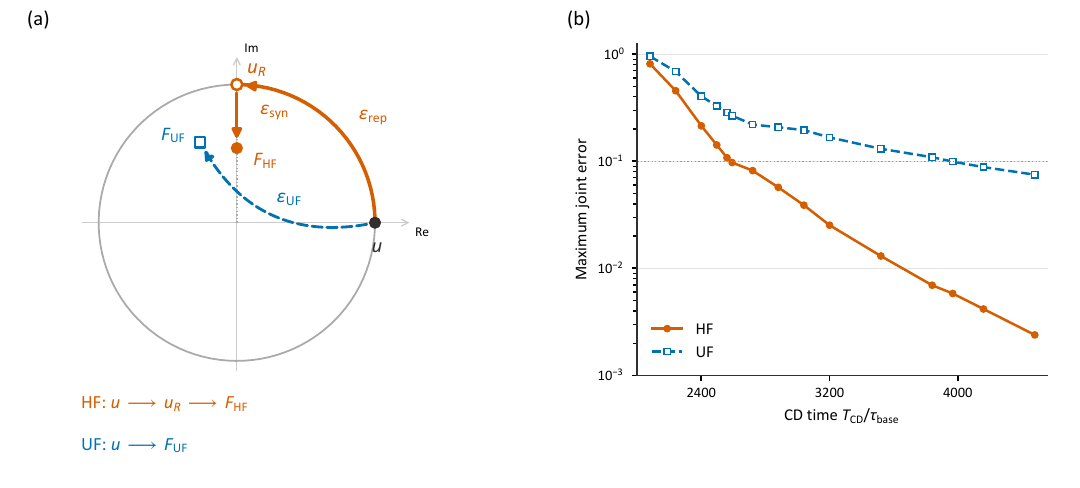}
\caption{Approximation stages and coherent serial cost. (a) Hamiltonian-first (HF) and unitary-first (UF) routes in the complex plane, with representation error $\eps_{\rm rep}$, HF synthesis error $\eps_{\rm syn}$, and UF approximation error $\eps_{\rm UF}$; zero HF synthesis phase error is shown schematically. (b) Best-found sampled joint-unitary error over all 48 block prefixes versus CD time for 16 repetitions of three distinct groups. Each point uses the same $513\times513$ coordinate grid and completion convention. The dotted line marks error $0.1$. These finite-search curves are not global optima or continuum error certificates.}
\label{fig:structure-programming}
\end{figure*}

For each candidate, we multiply the completed $2\times2$ matrices and maximize the joint-unitary error over all $48$ block prefixes on a common grid in $I^2$, where $I=\{q:V(q)\le1/2\}$. The plotted candidates first reach sampled error below $0.1$ at CD times $2592\tau_{\rm base}$ and $3968\tau_{\rm base}$ for HF and UF, respectively. The interval specifies coordinate support, not a quantum energy cutoff. This potential-only comparison includes neither kinetic evolution nor hardware noise. Appendix~\ref{app:fig3-serial-action} specifies the metric, cost convention, and numerical scope.

\paragraph{From queries to physical cost.}
For one scalar group with degree $K$ and base vector $\bm\nu$, the boundary lift preserves the CD pulse time
\begin{equation}
T_{\rm CD}=2K\,\tau_{\rm CD}(\bm\nu/2).
\label{eq:boundary-cd-time}
\end{equation}
A harmonic-$m$ query is charged $\tau_{\rm CD}(m\bm\nu/2)$, not one free balanced query.  In the fixed-drive model this is $|m|\tau_{\rm CD}(\bm\nu/2)$.  For $P\ne I$, Eq.~\eqref{eq:boundary-pauli-circuit} additionally uses two logical CNOTs, $C_P$, and $C_P^\dagger$.  Their total electronic-control duration is
\begin{equation}
T_{\rm el}=2\tau_{\rm CNOT}+\tau(C_P)+\tau(C_P^\dagger),
\label{eq:boundary-electronic-time}
\end{equation}
and is zero for $P=I$. Processing rotations, fixed basis and carrier changes, and per-pulse overhead add to $T_{\rm CD}+T_{\rm el}$. Virtual-$Z$ updates and active rotations need not have the same cost. Connectivity, routing, and calibration determine these contributions. Adjacent words with the same $P$ share boundary controls for every compiler family (Sec.~\ref{sec:controlled-lift}). Modewise displacements that commute need not be simultaneously executable under bandwidth, power, and crosstalk constraints.

For a CD drive rate $\Omega$, action $A_{\rm CD}$ lasts $\sqrt2\,A_{\rm CD}/\Omega$. One balanced query takes $\tau_{\rm base}:=\tau_{\rm CD}(\bm\nu/2)$. The serial-mode clock $\tau_{\rm CD}(\bm v)=\tau_{\rm unit}\|\bm v\|_1$ therefore gives $\tau_{\rm base}=\tau_{\rm unit}\|\bm\nu\|_1/2$. For equal modewise drive rates, $\tau_{\rm unit}=\sqrt2/\Omega$. The compiler benchmarks report $T_{\rm CD}/\tau_{\rm base}$, or $T_{\rm CD}/\tau_{\rm unit}$ when the ridge directions differ, whereas the noisy dynamics of Sec.~\ref{sec:short-window-dynamics} use a pulse-duration model in microseconds. The conversion of these reference times into device runtimes requires calibrated drive rates and control overheads. Appendix~\ref{app:processor-conventions} gives the rotation and carrier decompositions.

\section{Numerical benchmarks and noisy dynamics}
\label{sec:benchmarks}

We first compare compiler costs for supplied sinusoidal targets, then follow the simplified carbon dioxide model from coherent propagation to short-window noisy dynamics in Fig.~\ref{fig:fermi-grouped}.

\paragraph{Scope of the benchmarks.}
All benchmarks use electronic $P=I$. For a nonidentity $P$, the ideal boundary lift of Sec.~\ref{sec:controlled-lift} adds two CNOTs and one Clifford compute/uncompute pair. It preserves the paired-unitary error exactly. Scalar circuits therefore suffice for the compiler comparisons. The reported costs are the shortest circuits found by bounded searches, not global optima. They use a reference CD clock and are not calibrated device runtimes. Full-circle error checks use dense grids with analytic derivative padding. These floating-point checks are numerical evidence, not interval-arithmetic certificates. The benchmarks test the compiler, not the accuracy of a molecular potential energy surface. All costs depend on the supplied representation, error allocation, and outer step. Comparisons between molecular representations require the shared reference and accuracy contract of Appendix~\ref{app:working-space}.

\subsection*{Benchmark targets}
\label{sec:benchmark-targets}

\paragraph{Single tone.}
The single-tone target tests the onset of action savings over repeated weak steps:
\begin{equation}
\widehat U_a(\Lambda)=
\exp[-\ii\sigma_z^{(a)}\Lambda\cos(\pi\hat Q/L_Q+\phi_0)].
\label{eq:elementary-domain-target}
\end{equation}
Here $L_Q=10$ fixes the coordinate scale. The offset $\phi_0=0.025$ includes the offset-carrier operation in every family. One period spans the working coordinate window.

\paragraph{Dense harmonic profile.}
The dense profile tests the effect of grouping all harmonics through order eight:
\begin{equation}
h_s(\theta)=s\sum_{m=1}^{8}a_m\cos(m\theta),\qquad
 a_m=C\frac{(-1)^{m+1}}{m^2},
\label{eq:dense-eight-profile}
\end{equation}
Here $C\simeq1.98665$. The paired target is $\operatorname{diag}(\ee^{-\ii h_s},\ee^{\ii h_s})$. Molecular fits and mixed-parity profiles require further benchmarks.

\paragraph{Simplified carbon dioxide model.}
The two-mode target tests coherent reuse during vibrational energy exchange:
\begin{equation}
H_{\rm cm^{-1}}=\omega_s n_s+\omega_b n_b+\kappa R_{d_{\rm reg}}(Q_s,Q_b).
\label{eq:fermi-H}
\end{equation}
The first two terms give the harmonic Hamiltonian $H_0$. The function $R_{d_{\rm reg}}$ is the bounded sinusoidal replacement for $Q_sQ_b^2$ in Eq.~\eqref{eq:fermi-regularization}, with spacing $d_{\rm reg}=0.25$. Its three ridge directions, $d_{\rm reg}(1,\pm1)$ and $d_{\rm reg}(1,0)$, each contain two harmonics. The condition $\kappa\bra{0,2}R_{d_{\rm reg}}\ket{1,0}=|W_{\rm F}|$ fixes the resonant coupling. Appendix~\ref{app:grouped-benchmarks} gives the frequencies, detunings, and propagation convention.

\begin{figure*}[t]
\centering
\includegraphics[width=\textwidth]{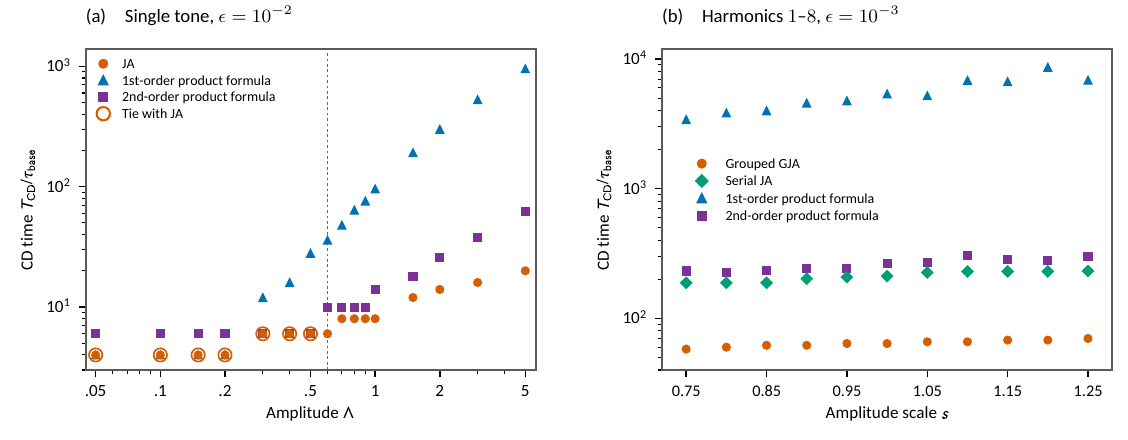}
\caption{Recovered compiler costs. (a) Single-tone target with coordinate scale $L_Q=10$ and offset $\phi_0=0.025$, common full-circle paired-unitary error tolerance $10^{-2}$. The vertical axis is the CD time in units of one balanced query, $\tau_{\rm base}$. Rings mark ties and the dashed guide marks the first strict saving on this grid. (b) Dense harmonics $1$--$8$ [Eq.~\eqref{eq:dense-eight-profile}], eleven scales, common full-circle paired-unitary error tolerance $10^{-3}$. Grouped GJA is compared with serial JA and first- and second-order product formulas (TGIFS). A harmonic-$m$ query costs $|m|\tau_{\rm base}$. Each panel uses one common criterion across its families. Rotation and electronic-control durations are excluded. These reference costs are not calibrated device runtimes. Appendix~\ref{app:common-metric} gives the selection and independent checks.}
\label{fig:compiler-benchmarks}
\end{figure*}

\begin{figure*}[t]
\centering
\includegraphics[width=\textwidth]{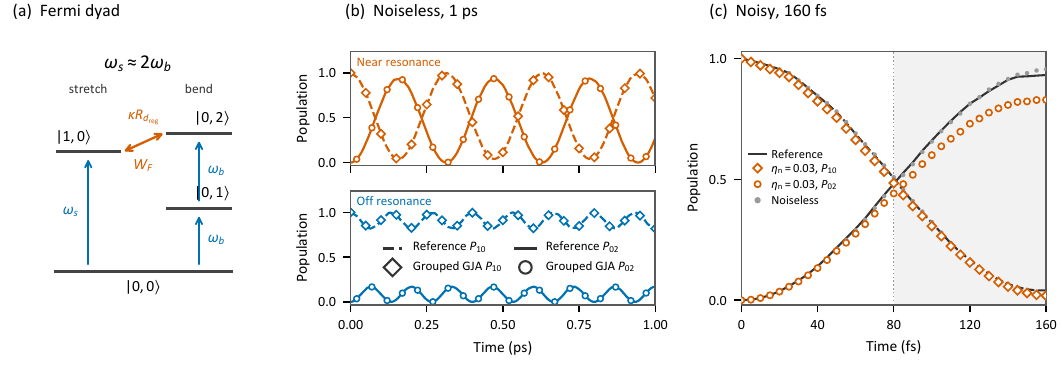}
\caption{Dynamics of the simplified carbon dioxide model from $\ket0_a\ket{1,0}$ with one persistent ancilla and no postselection or reset. (a) Fermi dyad with schematic levels. The bounded interaction $\kappa R_{d_{\rm reg}}$ couples $\ket{1,0}$ and $\ket{0,2}$ with $\kappa\bra{0,2}R_{d_{\rm reg}}\ket{1,0}=|W_{\rm F}|$, and the near-resonant detuning is $\delta=\omega_s-2\omega_b=-7.87$~cm$^{-1}$. (b) Noiseless $1$~ps populations near resonance (orange) and off resonance (blue, $\delta=-256.16$~cm$^{-1}$): reference lines and grouped-GJA markers. (c) Near-resonant populations of the rule-based $(3,3,5)$ circuit at noise-rate multiplier $\eta_{\rm n}=0.03$ over $160$~fs, with the noiseless circuit (gray) for comparison. Black lines show the reference populations. Noise acts throughout the $160$~fs. The circuit degrees were selected for the accuracy criteria over the first $80$~fs, and the shaded continuation shows the same circuit beyond that window. Panels (b) and (c) use $0.25$ and $5$~fs Strang steps, respectively. Noisy curves use $N_F=24$ Fock states per mode. Appendix~\ref{app:short-window-dynamics} gives the $N_F=32$ checks at $\eta_{\rm n}=0.03$, accuracy criteria, and pulse-clock assumptions.}
\label{fig:fermi-grouped}
\end{figure*}

\subsection{Single-ridge common-metric benchmark}
\label{sec:duration-repricing}
\label{sec:small-vs-finite}

TGIFS and JA use the same balanced query but choose their rotations differently. TGIFS repeats a fixed weak cell. Whole-response Fourier synthesis recovers a rotation sequence from a completed finite-angle response. JA and grouped GJA belong to the latter class. The all-forward identity in Appendix~\ref{smsec:tgifs-details} expresses $r$ basic TGIFS cells as $4r$ balanced queries and $2r$ equal rotations. This makes the shared architecture explicit.

The single-tone test asks where JA finite-angle synthesis begins to save displacement action relative to repeated weak steps. Integer circuit lengths can leave the methods tied at small amplitudes even when their asymptotic scalings differ.

We test the single-tone target of Eq.~\eqref{eq:elementary-domain-target} at sixteen amplitudes from $\Lambda=0.05$ to $5$. The sampling is denser below $1$. Appendix~\ref{app:elementary-duration} gives the circuit conventions. Appendix~\ref{app:common-metric} gives the full-grid searches.

We compare JA with first-order and second-order TGIFS under one criterion for all sixteen amplitudes. The paired-unitary error must be below $10^{-2}$ on the whole signal circle, the same metric as in the theorems and the other benchmarks. Appendix~\ref{app:common-metric} records the searches and the independent native-gate replay.

We report the CD time $T_{\rm CD}/\tau_{\rm base}$, which is $2K$ for JA, $4r$ for first-order TGIFS, and $4r+2$ for second-order TGIFS. Every balanced query has displacement magnitude $\pi/(2L_Q)$, and the offset-carrier time is proportional to $T_{\rm CD}$ by the same factor for all three families.

Figure~\ref{fig:compiler-benchmarks}(a) shows that JA ties the shorter TGIFS action at the seven sampled amplitudes up to $\Lambda=0.5$ and reduces it from $\Lambda=0.6$ onward. The action ratio reaches $20/62\simeq0.32$ at $\Lambda=5$. The onset of a sampled saving depends on the error metric and tolerance (Appendix~\ref{app:common-metric}).

For $\abs\Lambda\ll1$, the target admits the expansion $\widehat U_a(\Lambda)=I-\ii\sigma_z^{(a)}\Lambda\cos(\pi\hat Q/L_Q+\phi_0)+\calO(\Lambda^2)$. Both constructions reproduce this infinitesimal generator. This is consistent with TGIFS remaining competitive near the identity. At larger sampled amplitudes, TGIFS adds repetitions to suppress splitting error, whereas JA factors the finite-angle response directly.

\subsection{Multiharmonic and vibrational-dynamics benchmarks}
\label{sec:grouped-benchmarks}

We test grouping within a harmonic profile and coherent composition across the three ridges of the simplified carbon dioxide model. Figure~\ref{fig:compiler-benchmarks}(b) compares the harmonic compilers, including serial JA and both TGIFS orders. Figure~\ref{fig:fermi-grouped}(a,b) follows the molecular model through $1$~ps of noiseless propagation. Panel (c) examines its shorter, lower-precision noisy evolution. Appendix~\ref{app:grouped-benchmarks} gives the independent Direct Fourier checks, additional compiler comparisons, and numerical validation.

\paragraph{Grouping harmonics into one circuit.}
We test the dense profile of Eq.~\eqref{eq:dense-eight-profile} at eleven scales $s=0.75,0.80,\ldots,1.25$. The paired-unitary error tolerance is $10^{-3}$ on the full signal circle.

Figure~\ref{fig:compiler-benchmarks}(b) gives grouped CD times of $58$--$70\,\tau_{\rm base}$ versus $188$--$232\,\tau_{\rm base}$ for serial JA. At each scale, the reduction between the shortest circuits found for the two methods lies between $67.0\%$ and $71.3\%$. Integer circuit lengths and local searches can make these costs nonmonotone in amplitude. Independent Direct Fourier coefficient evaluation followed by the same grouped completion and recovery gives identical CD costs at all eleven scales. The reduction compares one grouped circuit with eight single-tone circuits run in sequence. It does not distinguish the two coefficient-evaluation routes. Appendix~\ref{app:dense-eight} specifies the searches and independent replay.

\paragraph{Fermi-resonant vibrational energy exchange.}
The bounded interaction of Eq.~\eqref{eq:fermi-H} models intramolecular vibrational energy redistribution (IVR) with the coupling and detuning of the carbon dioxide Fermi dyad~\cite{McCluskeyStoker2006CO2}. We propagate $\ket0_a\ket{1,0}$ for $1$~ps using $0.25$~fs Strang steps and a persistent synthesis ancilla. Figure~\ref{fig:fermi-grouped}(b) shows near-resonant energy exchange and its suppression upon detuning, with maximum reference populations $P_{02}=0.940$ and $0.171$, respectively. Couplings outside the calibrated dyad shift the exchange by a few percent, as checked in Appendix~\ref{app:grouped-benchmarks}.

We measure accuracy by the joint trace distance $D(t)=\tfrac12\|\rho(t)-\rho_{\rm ref}(t)\|_1$, with the reference ancilla in $\ket0$. The noisy dynamics below use the same metric. For these noiseless pure states we report its upper bound, the state distance $\|\Psi_{\rm comp}(t)-\ket0\psi_{\rm ref}(t)\|_2$ maximized on the $2$~fs output grid, which also includes the global phase. Accurate populations alone can hide coherent ancilla error: the loosest grouped-GJA candidate has near-resonant population error $0.00231$ but state distance $0.256$. This input-specific diagnostic does not bound an arbitrary-input channel.

The circuit in Fig.~\ref{fig:fermi-grouped}(b) keeps this distance, and hence $D$, below $10^{-2}$ for both detunings. Table~\ref{tab:fermi-compilers} and Appendix~\ref{app:grouped-benchmarks} give the compiler-cost comparison. We next include finite pulse durations and noise over a shorter interval.

\subsection{Short-window noisy dynamics}
\label{sec:short-window-dynamics}

Figure~\ref{fig:fermi-grouped}(c) follows the same initial state and unsplit reference Hamiltonian over $160$~fs. We select circuits on the first $80$~fs, which consist of sixteen $5$~fs Strang steps with three serial GJA groups and harmonic half-steps per step, and one ancilla is shared throughout. The compiled potential uses a doorway-matched representation with $d=0.75$, and its discrepancy from the $d_{\rm reg}=0.25$ reference is included in the final error. For the noisy states we evaluate $D(t)$ directly. The noiseless selection requires $D\le0.1$ and population, complex-coherence, and ancilla-leakage errors each at most $0.05$ on the $5$~fs output grid for both detunings.

We compare two circuits that pass these noiseless criteria. The baseline uses group degrees $(K_+,K_-,K_s)=(4,4,5)$ with default completion. The rule-based circuit uses $(3,3,5)$ with default completion and the step-alternating completion phase of Sec.~\ref{sec:composition}, without any search over roots or phases. Default completion alone fails at $(3,3,5)$ with $D=0.112$, and no cheaper allocation passes with the rule (Appendix~\ref{app:short-window-dynamics}). The group order is common to both circuits.

The baseline contains $672$ modewise CD pulses and $432$ rotations, and the rule-based circuit contains $544$ and $368$. Their modeled durations are $188.8$ and $153.3~\mu\mathrm{s}$, so the rule shortens the circuit by $18.8\%$. Both circuits use the same local Pauli-frame selection rule, absorbed into existing rotations and signed CD pulses with identity boundary frames.

Relaxation, pure dephasing, and mode loss act during every pulse, including finite-duration rotations. The noise-rate multiplier $\eta_{\rm n}$ scales the rates of illustrative circuit-QED lifetimes $T_1=T_\phi=60~\mu\mathrm{s}$ and $T_{\rm cav}=160~\mu\mathrm{s}$, which lie within the ranges reported for superconducting bosonic-QSP control~\cite{Fontaine2026Programming}. The noisy criterion allows $D\le0.2$ and keeps the other three thresholds at $0.05$ on the same output grid.

We summarize each circuit by its final joint-state fidelity,
\begin{equation}
F=\bra{0,\psi_{\rm ref}(80~\mathrm{fs})}\rho(80~\mathrm{fs})\ket{0,\psi_{\rm ref}(80~\mathrm{fs})}.
\label{eq:short-window-fidelity}
\end{equation}
This overlap includes the ancilla and all representation, synthesis, propagation, and noise errors. Without noise, the two near-resonant fidelities are $0.9921$ and $0.9955$, so the rule also reduces the coherent serial error.

At $\eta_{\rm n}=0.03$, the baseline and rule-based circuits give $F=0.9289$ and $0.9331$. At $\eta_{\rm n}=0.1$, they give $0.7976$ and $0.8037$, so the rule-based circuit keeps a higher fidelity at every tested rate (Table~\ref{tab:staged-ablation}). When the same circuit continues to $160$~fs [Fig.~\ref{fig:fermi-grouped}(c)], the same overlap evaluated at $160$~fs falls from $0.933$ to $0.854$ at $\eta_{\rm n}=0.03$. The noiseless circuit keeps $0.989$, so the additional loss comes mainly from noise. The transferred population then reaches $P_{02}=0.83$ instead of the reference $0.93$. Over the full $160$~fs the same circuit no longer meets the accuracy criteria. At $\eta_{\rm n}=0.03$ its trace distance stays below $0.2$ ($0.186$), but its population and coherence errors reach $0.104$ and $0.094$. Without noise the coherence error already reaches $0.077$ and the trace distance $0.103$, so a longer window requires a new degree allocation. These $160$~fs results use $N_F=24$ Fock states per mode. Their mixed-state difference from $N_F=32$ grows from $0.003$ at $80$~fs to $0.008$ at $160$~fs, above the $0.005$ cutoff target but small compared with these errors.

Both circuits narrowly exceed the coherence-error threshold at $\eta_{\rm n}=0.03$, with near-resonant errors $0.0521$ and $0.05004$, while their trace-distance, population, and leakage errors meet the noisy thresholds. Appendix~\ref{app:short-window-dynamics} reports the cutoff checks.

At nominal rates the two fidelities fall to $0.1586$ and $0.1688$, and both fail the accuracy criteria. The low-noise point $\eta_{\rm n}=0.03$ corresponds to $T_1=T_\phi=2$~ms and $T_{\rm cav}=5.33$~ms under the illustrative model. These calculations concern the specified input and effective pulse model. Hardware assessment requires internal ECD dynamics, calibrated envelopes, state preparation, readout, and their overheads.

\section{Discussion and conclusion}
\label{sec:conclusion}

Our serial GJA compiler turns real sinusoidal models of molecular interactions into native oscillator--qubit circuits by completing and factoring each compatible group's phase response. The real potential model preserves unit modulus before synthesis, separating potential error from synthesis-induced ancilla leakage and allowing the model to be reused across evolution times. Completed circuits share one ancilla without intermediate measurement or reset. GJA synthesis implements multiharmonic sinusoidal vibronic interactions at a cost that grows near-optimally with the interaction strength and the precision. Its cost matches our lower bounds up to a constant factor for a single sinusoidal term and to leading order in the interaction strength for a group of harmonics.

The benchmarks test finite-circuit costs and coherent and noisy dynamics under their stated error criteria. The serial HF--UF comparison illustrates joint selection of potential and gate degrees, while the multiharmonic benchmark isolates savings from grouping compatible tones. Among the programmable constructions compared in Fig.~\ref{fig:compiler-benchmarks} and Table~\ref{tab:fermi-compilers}, GJA gives the shortest or equal circuits in every tested case. These finite searches provide achievable costs rather than global optima. In the simplified carbon dioxide model, one persistent ancilla supports 4000 time steps over one picosecond, with joint-state error below $10^{-2}$ on the output grid for both detunings. Over a shorter $80$~fs window with pulse-level noise, alternating the completion phase between steps admits a lower degree allocation. Under common input-specific ideal criteria, it shortens the modeled circuit by $18.8\%$ and raises the fidelity at every tested noise rate. The noisy dynamics approach the accuracy criteria only at reduced noise rates, which identifies a conditional operating point rather than demonstrated hardware feasibility.

These results concern phase-gate compilation under exact signal access and a modeled pulse clock, rather than an end-to-end molecular simulation on hardware. Once validated potentials, state preparation, and measurement complete the workflow, the noise level identified here serves as a concrete target for simulating anharmonic multimode molecular quantum dynamics on near-term oscillator--qubit processors.

\begin{acknowledgments}
This work was supported by the National Research Foundation of Korea (NRF) funded by the Ministry of Science and ICT (RS-2023-NR068116, RS-2025-03532992), by the Korea Health Industry Development Institute (KHIDI) funded by the Ministry of Health and Welfare through the Korean ARPA-H Project (RS-2025-25456722), by the Institute for Information \& Communications Technology Promotion (IITP) funded by the Korean government (MSIP) (No. 2019-0-00003), and by the Yonsei University Research Fund (2025-22-0140).

AI tools from OpenAI and Anthropic assisted with construction details from our intuitions, theoretical analysis, numerical simulations, figure preparation, and manuscript editing. The authors reviewed the outputs and take responsibility for the final manuscript.
\end{acknowledgments}

\appendix

\section{Coherent compiler construction}
\label{smsec:ja-details}

This appendix gives the coefficient tails, coherent completion, and phase-only construction used by Theorem~\ref{thm:jaoqgqsp-trig}.

\subsection{Single-tone truncation bounds}

The Jacobi--Anger tail and its joint angle--precision inversion are standard consequences of Lemmas 57 and 59 in the arXiv version of Ref.~\cite{GilyenSuLowWiebe2019}.  We state them here without reproducing that proof.

Let $S^{c}_{K,\Lambda}$ be the cosine response of Eq.~\eqref{eq:ja-raw-cosine-response}, and let $S^{s}_{K,\Lambda}(z)=\sum_{\abs n\le K}(-1)^nJ_n(\Lambda)z^n$ be its sine counterpart.
\paragraph{Jacobi--Anger tail}
\label{prop:ja-tail}
For real $\Lambda$ and $K\ge0$,
\begin{align}
\sup_{\theta\in\bbR}\abs{\ee^{-\ii\Lambda\cos\theta}-S^{c}_{K,\Lambda}(\ee^{\ii\theta})}
&\le 2\sum_{n>K}\abs{J_n(\abs\Lambda)},\nonumber\\
\sup_{\theta\in\bbR}\abs{\ee^{-\ii\Lambda\sin\theta}-S^s_{K,\Lambda}(\ee^{\ii\theta})}
&\le 2\sum_{n>K}\abs{J_n(\abs\Lambda)}.
\label{eq:bessel-tail}
\end{align}
The inequality $\abs{J_n(\Lambda)}\le (\abs\Lambda/2)^n/n!$ then shows that, for every $0<\eps<1/2$, there is a truncation order
\begin{equation}
K=\calO\!\left(
1+\abs\Lambda+
\frac{\log(1/\eps)}
{\log\!\left(\mathrm e+\log(1/\eps)/\abs\Lambda\right)}
\right)
\label{eq:ja-degree}
\end{equation}
such that both truncation errors in Eq.~\eqref{eq:bessel-tail} are at most $\eps$, for $\Lambda\ne0$. At $\Lambda=0$ the response is exactly one and $K=0$, so the expression containing $1/|\Lambda|$ is not used.

The asymptotic form of the degree matches that of block-encoded Hamiltonian simulation~\cite{GilyenSuLowWiebe2019}.  The oracle lower bounds in that setting are derived from sparse-Hamiltonian reductions~\cite{BerryAhokasCleveSanders2007,BerryChildsCleveKothariSomma2014}.  The lower bound used here instead applies in the declared unitary-signal model.  That model was recently characterized by adversary bounds~\cite{Laneve2026AdversaryQSP}.

\subsection{Paired completion and coherent composition}
\begin{lemma}[From selected response to paired unitary]
\label{cor:ja-full-joint-completion}
Let $U(F,H)$ have the unitary block form of Eq.~\eqref{eq:ja-operator-block-form}, and let $F,H,V$ be functions of the same normal signal, with $V$ unitary and $\|F-V\|\le\eta<1$.
Then
\begin{align}
\|H\|&\le\sqrt{2\eta-\eta^2},\label{eq:ja-complement-from-selected}\\
\|U(F,H)-\widehat V\|^2
&=\|2I-F^\dagger V-V^\dagger F\|\le2\eta.
\label{eq:paired-paraunitary-exact-distance}
\end{align}
\end{lemma}
\begin{proof}
The selected error gives $F^\dagger F\ge(1-\eta)^2I$, so unitarity bounds the complement. In $(U-\widehat V)^\dagger(U-\widehat V)$, commutativity cancels the off-diagonal terms and identifies both diagonal blocks with $2I-F^\dagger V-V^\dagger F$. This expression equals $(V-F)^\dagger V+V^\dagger(V-F)$, which proves the bound.
\end{proof}

For any ordered unitary words, telescoping gives
\begin{equation}
\|U_B\cdots U_1-\widehat V_B\cdots\widehat V_1\|
\le\sum_b\|U_b-\widehat V_b\|.
\label{eq:ja-full-joint-composition}
\end{equation}
The words may therefore share the ancilla without measurement. At fixed $F$, every valid spectral complement has the same pointwise magnitude and the same paired-unitary error.

\subsection{Proof of the single-tone upper bound}

\begin{proof}[Single-tone upper bound in Theorem~\ref{thm:jaoqgqsp-trig}]
Take the raw tail tolerance $\delta=\eps^2/16$ in Eq.~\eqref{eq:ja-degree}. This gives a truncated sine response $A_0+\ii C_0$ of the stated degree, with $A_0$ an even cosine series and $C_0$ an odd sine series. Division by $1+\delta$ gives $A_1+\ii C_1$ with selected error at most $2\delta$, and Theorem~2 of Ref.~\cite{LowChuang2017} supplies the remaining real components $B_1,D_1$. The rotation $A_2=A_1\cos\rho+B_1\sin\rho$ with $\cos\rho=A_1(0)$ and $\sin\rho=B_1(0)$ sets $A_2(0)=1$ and keeps the parity and the maximum harmonic. Since $|B_1|\le2\sqrt\delta/(1+\delta)$, it changes $A_1$ by at most $6\delta$, so $F=A_2+\ii C_1$ has selected error at most $8\delta=\eps^2/2$. These components satisfy the achievable-pair conditions $A_2^2+C_1^2\le1$, $A_2(0)=1$, and the Fourier parities, so Theorem~1 of that reference gives $2K$ equatorial signal rotations with phase-only control, and fixed input and output Hadamards map its $\ket+$ block to the paired $\sigma_z$ basis without changing the CD count. The word applied to $P\otimes\exp[\ii(\bmu\cdot\bQ+\phi I)]$ gives the sine target, and applied to $\ii P\otimes\exp[\ii(\bmu\cdot\bQ+\phi I)]$ it gives the cosine target. The remaining frame $P^K$, the sign $(-1)^K$ on the $P=-1$ sector, is removed by one $P^{K\bmod2}$ boundary correction. All blocks are functions of the same normal joint signal, so Lemma~\ref{cor:ja-full-joint-completion} gives paired-unitary error at most $\eps$, and $\log(1/\delta)=2\log(1/\eps)+\log16$ gives Eq.~\eqref{eq:jaoqgqsp-query}.
\end{proof}

\subsection{Generalized Jacobi--Anger coefficients and one-signal synthesis}
\label{smsec:generalized-ja}

The generalized Jacobi--Anger identity is classical~\cite{DattoliEtAl1992GeneralizedBessel,DattoliTorreLorenzutta1998,KuklinskiHague2021}.  Throughout this subsection, the tone list is finite, amplitudes and offsets are real, and the quadratures strongly commute.

\paragraph{Classical expansion and frequency grouping.}
For arbitrary real ridge vectors $\bmu_r$, multiplication of the ordinary Jacobi--Anger series gives
\begin{align}
&\exp\!\left[-\ii\sum_{r=1}^{R}a_r
\sin(\bmu_r\cdot\bm q+\phi_r)\right]\nonumber\\
&\quad=\sum_{\bm k\in\mathbb Z^R}
\left[\prod_{r=1}^{R}J_{k_r}(-a_r)\ee^{\ii k_r\phi_r}\right]
\ee^{\ii(\sum_r k_r\bmu_r)\cdot\bm q}.
\label{eq:generalized-ja-arbitrary-ridges}
\end{align}
Since $\sum_{k\in\mathbb Z}|J_k(a)|<\infty$ for every finite real $a$, the series converges absolutely and uniformly on real $\bm q$ and may be regrouped. The choice $\bmu_r=\omega_r\bm\xi$ gives Eq.~\eqref{eq:generalized-ja-multitone}. If $\bmu_r=m_r\bm\nu$ with integers $m_r$ and $\bm\nu=\beta\bm\xi$, collecting $n=\sum_rm_rk_r$ gives Eqs.~\eqref{eq:generalized-ja-response}--\eqref{eq:generalized-ja-coefficients}. Negative $m_r$ and constant terms are absorbed into amplitudes, offsets, and separate phases, so we take $1\le m_r\le L_{\rm h}$.

\begin{lemma}[Finite-harmonic generalized Jacobi--Anger tail]
\label{lem:generalized-ja-tail}
Let $G(\ee^{\ii\theta})=\exp[-\ii h(\theta)]$, where
$h(\theta)=\sum_{r=1}^{R}a_r\sin(m_r\theta+\phi_r)$,
and let $C_n$ be its Laurent coefficients.  With $L_{\rm h}=\max_rm_r$, define
\begin{equation}
S=\sum_r m_r|a_r|,\qquad
B(\rho)=\sum_r|a_r|\sinh(m_r\rho).
\end{equation}
For every $\rho>0$ and $K\ge0$,
\begin{align}
|C_n|&\le\exp[B(\rho)-|n|\rho],\nonumber\\
\sum_{|n|>K}|C_n|
&\le\frac{2\exp[B(\rho)-(K+1)\rho]}{1-\ee^{-\rho}}.
\label{eq:generalized-ja-contour-tail}
\end{align}
For $S>0$ and $0<\delta<1/2$, a sufficient degree for this tail to be at most $\delta$ is
\begin{equation}
K=\calO\!\left(S+
\frac{L_{\rm h} T_\delta}{\log(\mathrm e+L_{\rm h} T_\delta/S)}\right),
\quad T_\delta=\log(4L_{\rm h}/\delta).
\label{eq:generalized-ja-tail-degree}
\end{equation}
If $S=0$, the response is the identity and $K=0$.
\end{lemma}

\begin{proof}
The finite trigonometric polynomial is entire in $\theta$.  For real $u$,
$|\operatorname{Im} h(u\pm\ii\rho)|\le B(\rho)$.
Shift the Fourier-coefficient contour by $-\ii\rho$ for $n>0$ and by $+\ii\rho$ for $n<0$.  Periodicity cancels the vertical sides and yields the coefficient bound, which is the strip bound of Ref.~\cite{HongOQGQSP2026} applied to a sinusoidal profile.  The sum of the two geometric tails yields Eq.~\eqref{eq:generalized-ja-contour-tail}.

Convexity of $\sinh$ on the positive real axis gives
$B(\rho)\le(S/L_{\rm h})\sinh(L_{\rm h}\rho)$.  Put
$a=S/L_{\rm h}$, $t=L_{\rm h}\rho=\log(\mathrm e+T_\delta/a)\ge1$.
Since $1-\ee^{-t/L_{\rm h}}\ge1-\ee^{-1/L_{\rm h}}\ge1/(2L_{\rm h})$,
the tail is at most
\begin{equation}
4L_{\rm h}\exp\!\left[a\sinh t-(K+1)t/L_{\rm h}\right].
\end{equation}
Also $a\sinh t\le(\mathrm e a+T_\delta)/2$.
The choice $K=\lceil2S+2L_{\rm h} T_\delta/t\rceil$ makes the exponent at most $-T_\delta$, because $t\ge1$ and $2>\mathrm e/2$.  This proves Eq.~\eqref{eq:generalized-ja-tail-degree}.
\end{proof}

\paragraph{Explicit superexponential regime.}
For $S>0$, put $n_*=|n|>S/2$. Then
\begin{equation}
|C_n|\le\left(\frac{\mathrm e S}{2n_*}\right)^{n_*/L_{\rm h}}.
\label{eq:generalized-ja-superexponential-main}
\end{equation}

To obtain this bound, choose
$\rho=L_{\rm h}^{-1}\log(2n_*/S)>0$ in Eq.~\eqref{eq:generalized-ja-contour-tail}.
The estimate $B(\rho)\le(S/L_{\rm h})\sinh(L_{\rm h}\rho)\le(S/2L_{\rm h})\ee^{L_{\rm h}\rho}=n_*/L_{\rm h}$
gives Eq.~\eqref{eq:generalized-ja-superexponential-main}.
Likewise, setting $N=K+1>S/2$ gives the explicit uniform tail majorant
\begin{equation}
\tau_K\le
\frac{2}{1-(S/2N)^{1/L_{\rm h}}}
\left(\frac{\mathrm e S}{2N}\right)^{N/L_{\rm h}},
\label{eq:generalized-ja-superexponential-tail}
\end{equation}
These are upper bounds, not coefficientwise asymptotic equalities. Symmetry, a nonprimitive frequency lattice, or cancellation can make individual coefficients vanish, and the rate is not uniform when amplitudes or the maximum harmonic grow with $K$.

\paragraph{One completion and one inverse factorization.}
\emph{Step 1: contractive selected response.} For the response $G$ of Eq.~\eqref{eq:generalized-ja-response}, let $T_K(z)=\sum_{|n|\le K}C_nz^n$ and let $\tau_K$ bound the omitted absolute tail. Then $|T_K|\le1+\tau_K$ on the unit circle, and $F_K=T_K/(1+\tau_K)$ obeys $|F_K|\le1$ and $\sup_{|z|=1}|F_K-G|\le2\tau_K$.
\emph{Step 2: unitary completion.} The nonnegative trigonometric polynomial $1-|F_K|^2$ has degree at most $2K$.  Fej\'er--Riesz factorization gives a polynomial $q_K$ of degree at most $2K$ with $|q_K|^2=1-|F_K|^2$.  Thus $H_K(z)=z^{-K}q_K(z)$ has support in $[-K,K]$ and completes $F_K$ to the paired determinant-one Laurent matrix
\begin{equation}
\mathcal U_K(z)=
\begin{pmatrix}
F_K(z)&-H_K(z)^*\\
H_K(z)&F_K(z)^*
\end{pmatrix},\qquad |z|=1.
\label{eq:generalized-ja-paired-matrix}
\end{equation}
Here $p^*(z):=\overline{p(1/\bar z)}$ denotes Laurent para-conjugation, equal to pointwise conjugation on the circle.  

\emph{Step 3: factor into balanced queries.} Apply the full-matrix Laurent-factorization theorem of Ref.~\cite{Haah2019QSP} to the even loop $\mathcal U_K(t^2)$, where $t=\ee^{\ii\theta/2}$.  Its degree is at most $2K$. It therefore factors into a constant $\mathrm{SU}(2)$ matrix and at most $2K$ primitives $t\Pi_j+t^{-1}(I-\Pi_j)$. Here $\Pi_j$ are rank-one ancilla projectors.  A constant rotation that diagonalizes each projector converts each primitive to $D(t)=\operatorname{diag}(t,t^{-1})$.  This factors the chosen full matrix, including its complement, rather than only the selected response.  Its even degree gives an even query count.  The identity $D(t)(\ii X)D(t)(-\ii X)=I$ permits padding to $2K$ without a residual phase.  The resulting word is
\begin{equation}
\begin{split}
\mathcal U_K(W_{\bm\nu})&=R_{2K}D R_{2K-1}\cdots D R_0,\\
D&=\operatorname{diag}(W_{\bm\nu}^{1/2},W_{\bm\nu}^{-1/2}),
\end{split}
\label{eq:generalized-ja-word}
\end{equation}
with constant ancilla rotations $R_j\in\mathrm{SU}(2)$. This existence statement, independent of the numerical recovery backend, allows generic rotations. Phase-only recovery imposes additional symmetry and endpoint conditions (Appendix~\ref{app:processor-conventions}). A fixed basis change gives the native $X$-basis conditional displacements.  The half-powers here are the physically specified displacements $\exp(\pm\ii\bm\nu\cdot\bQ/2)$.

Functional calculus transfers the uniform scalar bounds to operator norm for $W_{\bm\nu}$.  By Lemma~\ref{cor:ja-full-joint-completion}, the distance from Eq.~\eqref{eq:generalized-ja-paired-matrix} to $\operatorname{diag}(G,G^*)$ is at most $2\sqrt{\tau_K}$.  The choice $\tau_K\le\eps^2/4$ proves the grouped scalar upper bound in Theorem~\ref{thm:jaoqgqsp-trig}. Conjugation by the boundary lift of Eq.~\eqref{eq:boundary-pauli-circuit} inserts the common electronic $P$ and preserves both the operator-norm error and the scalar query count.  Numerical recovery additionally budgets coefficient, contractivity, completion, and inverse-factorization errors as described in Appendix~\ref{app:angle-recipe}.

\paragraph{Global multi-index GJA and partially serial synthesis.}\label{app:multidirection-response}\label{app:multivariable-ja-tail}
Equation~\eqref{eq:generalized-ja-arbitrary-ridges} describes the complete nonseparable response. On a common frequency lattice, collecting equal vector frequencies gives global multi-index Fourier coefficients $C_{\bm n}$. These scalar coefficients and their truncation tails do not themselves supply a multivariable unitary completion or its native factorization. We leave direct global completion and recovery, with guaranteed approximation error and CD action, as an open extension. The constructive decision algorithm of Ref.~\cite{ItoEtAl2026MQSPDecision} tests a supplied polynomial pair within a specified multivariable QSP model and factors it when implementable. It does not construct the pair from scalar GJA coefficients.

Partially serial GJA partitions the tones into commensurate groups with a common $P$, completes each group as above, and concatenates the completed words on the persistent ancilla. The paired errors add by unitary telescoping (Sec.~\ref{sec:resources}), and the representation error adds as in Eq.~\eqref{eq:representation-leakage-separation}. For two completed groups in the convention of Eq.~\eqref{eq:generalized-ja-paired-matrix}, $\mathcal U_{\rm tot}=\mathcal U_2\mathcal U_1$ extends Eq.~\eqref{eq:serial-completion-interference} to the complementary block,
\begin{equation}
\begin{split}
F_{\rm tot}&=F_2F_1-H_2^*H_1,\\
H_{\rm tot}&=H_2F_1+F_2^*H_1 .
\end{split}
\label{eq:serial-complement-interference}
\end{equation}
Complementary phases can therefore change the final leakage at fixed group responses. Concatenation does not by itself reduce CD action, and serial and global synthesis have no universal leakage ordering.

\section{Query and displacement-action bounds}
\label{app:proofs}
This appendix proves the single-tone lower bound of Theorem~\ref{thm:jaoqgqsp-trig} and the bounds of Theorem~\ref{prop:action-optimality}. The latter separates a fixed base-query model from a single-ancilla real-kick model, and each proof keeps its own assumptions.

\subsection{Single-tone query lower bound}
The processor class is that of Theorem~\ref{thm:jaoqgqsp-trig}: each nonsignal gate acts on finite signal-processing ancillas and is independent of the signal, and in the joint-signal class these gates preserve every eigenspace of $P$. A power $W^m$ is charged $|m|$ calls. Since $\mathrm{diag}(z^{1/2},z^{-1/2})=z^{-1/2}\,\mathrm{diag}(z,1)$, a word with $2K$ balanced queries has the same $2K+1$ Laurent coefficients as one with $2K$ controlled calls, so the upper and lower bounds count the same resource. The fixed unnormalized selected response $F_Q$ satisfies $\sup_{\theta}|F_Q(e^{\ii\theta})-e^{-\ii\Lambda\cos\theta}|\le\eps$. Adaptive measurement, normalized postselection, a finite signal spectrum, and free power queries define different models.

\begin{proof}[Single-tone lower bound in Theorem~\ref{prop:ja-single-ridge-lower}]
For the Pauli-conditioned transfer, choose any nonempty eigenspace of $P$ with eigenvalue $p\in\{+1,-1\}$. Every non-query gate preserves this sector. The joint signal of Eq.~\eqref{eq:joint-pauli-signal} acts there as $pW_{\bmu,\phi}$. For $\bmu\ne\bm0$, its spectrum is the full unit circle. For $p=-1$, shifting the signal angle by $\pi$ converts $p\cos\theta$ to the scalar cosine target without changing Laurent degree or the uniform error. It therefore suffices to prove the scalar case.
Every matrix element of a signal-independent coherent control is constant in $z$. One controlled query has Laurent degree at most one.  Induction
over the queries therefore gives
\begin{equation}
F_Q(z)=\sum_{k=-Q}^{Q}f_kz^k.
\label{eq:ja-lower-laurent-form}
\end{equation}
This remains true with finite work ancillas and coherent control by stored branch labels because the selected response is a fixed matrix element of the resulting block product.

Consider the cosine target. The sine target follows by $F_Q(z)\mapsto F_Q(\ii z)$, which preserves degree.  The symmetrization of Eq.~\eqref{eq:ja-lower-laurent-form} under $z\mapsto z^{-1}$ does not increase the uniform error because the target is invariant under this map.  Since $z^k+z^{-k}=2T_k[(z+z^{-1})/2]$, the symmetrized response has the form $p_Q(\cos\theta)$ for a complex algebraic polynomial of degree at most $Q$. Thus
\begin{equation}
\sup_{x\in[-1,1]}\abs{p_Q(x)-\ee^{-\ii\Lambda x}}\le\eps.
\label{eq:ja-lower-algebraic-approx}
\end{equation}
For any integer $n>Q$, the $n$th forward difference of $p_Q$ on the grid
$-1,-1+2/n,\ldots,1$ vanishes, whereas
\begin{equation}
\abs{\Delta_{2/n}^n\ee^{-\ii\Lambda x}\big|_{x=-1}}
=2^n\abs{\sin(\Lambda/n)}^n.
\end{equation}
The binomial sum of the grid errors in Eq.~\eqref{eq:ja-lower-algebraic-approx} is at most $2^n\eps$.  Consequently,
\begin{equation}
\eps\ge\abs{\sin(\abs\Lambda/n)}^n
\qquad(n>Q).
\label{eq:ja-lower-finite-difference}
\end{equation}

The phase amplitude gives a second obstruction.  At the points
$(\pi/2+j\pi)/\abs\Lambda$, $j\in\mathbb Z$, that lie in $[-1,1]$, the imaginary part of the
target in Eq.~\eqref{eq:ja-lower-algebraic-approx} alternates between $1$ and
$-1$ up to an overall sign.  Since $\eps<1/2$, the real polynomial
$\operatorname{Im}p_Q$ has alternating signs at these points.  The number of intervening zeros implies
\begin{equation}
Q\ge\left\lfloor\frac{2\abs\Lambda}{\pi}\right\rfloor-1
\ge\frac{2\abs\Lambda}{\pi}-2.
\label{eq:ja-lower-sign-change}
\end{equation}

To combine the obstructions uniformly, put $a=|\Lambda|$, $\ell=\log(1/\eps)$, $d=\log(\mathrm e+\ell/a)$, and $b=\ell/d$. We claim
\begin{equation}
Q+1\ge(1+a+b)/58.
\label{eq:ja-lower-explicit-constant}
\end{equation}
If $a+b<57$, this follows from $Q+1\ge1$. Otherwise, when $b<8a$, Eq.~\eqref{eq:ja-lower-sign-change} gives $Q+1\ge2a/\pi-1>2(a+b)/(9\pi)-1\ge(1+a+b)/58$.
When $b\ge8a$, choose $n=\lfloor b/4\rfloor$. Since $a+b\ge57$ and $a\le b/8$, we have $b\ge152/3>8$, so
$n\ge b/8\ge a$ and $n\ge(a+b)/9$.
Also $n\le\ell/(4d)$ and $\log(2n/a)<d$, hence $n\log(2n/a)<\ell/4$. The inequality $\sin x\ge x/2$ on $[0,1]$ yields
$|\sin(a/n)|^n\ge(a/2n)^n>e^{-\ell/4}>\eps$.
If $Q<n$, this contradicts Eq.~\eqref{eq:ja-lower-finite-difference}. Otherwise $Q\ge n\ge(a+b)/9$ proves the claim.

\end{proof}

\subsection{Grouped-profile amplitude lower bound}
\label{app:grouped-amplitude-proof}

\begin{proof}[Proof of Theorem~\ref{prop:grouped-amplitude-optimality}(i)]
We first establish a degree obstruction.  Let $p(\ee^{\ii\theta})$ be a Laurent polynomial of degree at most $K$ that approximates $\ee^{-\ii\lambda h(\theta)}$ uniformly within $\eps<1$.  We claim
\begin{equation}
K+1\ge\frac{\lambda\Delta_h}{2\pi}.
\label{eq:grouped-oscillation-degree}
\end{equation}
If $\lambda\Delta_h\le2\pi$, the claim follows from $K\ge0$.  Otherwise choose a minimum and a maximum of $h$ and connect them by an oriented interval of length less than $2\pi$.  The levels
\begin{equation}
y_j=\frac{\pi/2+j\pi}{\lambda},\qquad j\in\mathbb Z,
\label{eq:grouped-sign-levels}
\end{equation}
that lie in $[\min h,\max h]$ number at least $\lambda\Delta_h/\pi-1$.  Begin at the minimum and select the first point where each successive level is attained.  Continuity orders these points along the interval even when $h$ is not monotone.  At them, the imaginary part of the target alternates between $+1$ and $-1$.  Since $\eps<1$, $\operatorname{Im}p$ has the same alternating signs and hence at least $\lambda\Delta_h/\pi-2$ distinct zeros between these points.

The function $\operatorname{Im}p$ is a real trigonometric polynomial of degree at most $K$.  It is not identically zero because it takes nonzero values at the selected points.  Multiplication of its Laurent representation by $z^K$ gives a nonzero algebraic polynomial of degree at most $2K$. It therefore has at most $2K$ distinct zeros on one period.  Thus
$2K\ge\lambda\Delta_h/\pi-2$, proving Eq.~\eqref{eq:grouped-oscillation-degree}. Harmonic cancellation is permitted and enters only through $\Delta_h>0$.

To pass to the balanced-query count, set $w=\ee^{\ii\theta/2}$.  Each query has Laurent degree one in $w$, including controlled calls, and every nonsignal control is constant in $w$.  Consequently, the fixed selected matrix element of an $N$-query circuit has degree at most $N$.  Average its values at $\theta$ and $\theta+2\pi$:
\begin{equation}
F_N^{\rm even}(\theta)=\tfrac12[F_N(\theta)+F_N(\theta+2\pi)].
\label{eq:grouped-even-response}
\end{equation}
This average is a proof device, not an additional circuit operation.  The target is $2\pi$-periodic, so the average retains the same error bound.  The shift sends $w$ to $-w$, removing the odd powers.  The result is therefore a Laurent polynomial in $z=w^2$ of degree at most $\lfloor N/2\rfloor$.  Equation~\eqref{eq:grouped-oscillation-degree} gives
\begin{equation}
\frac N2+1\ge\lfloor N/2\rfloor+1
\ge\frac{\lambda\Delta_h}{2\pi},
\end{equation}
which proves Eq.~\eqref{eq:grouped-amplitude-lower}.  The reduction uses the bound for all real $\theta$, including both lifts of the balanced signal.  A paired-unitary approximation controls its selected matrix element and thus obeys the same lower bound.

For the upper bound, write the fixed profile as
$h(\theta)=h_0+\sum_rb_r\sin(m_r\theta+\phi_r)$,
with $1\le m_r\le L_{\rm h}$.  Cosine terms are included by shifting their offsets.  The constant $h_0$ is implemented by a signal-independent ancilla $Z$ rotation.  Since $h$ is nonconstant,
$S_h=\sum_rm_r|b_r|>0$.
For the target $\lambda h$, the grouped construction replaces $S$ in Eq.~\eqref{eq:generalized-ja-degree} by $\lambda S_h$, yielding
\begin{equation}
\begin{split}
N_{\rm GJA}=2K
&=\calO\!\left(1+\lambda S_h+
\frac{L_{\rm h} T}{\log(\mathrm e+L_{\rm h} T/(\lambda S_h))}\right),\\
T&=\log(16L_{\rm h}/\eps^2).
\end{split}
\label{eq:grouped-amplitude-upper}
\end{equation}
At fixed $h$ and $\eps$, this is $\calO_{h,\eps}(\lambda)$ as $\lambda\to\infty$.  This estimate and Eq.~\eqref{eq:grouped-amplitude-lower} together prove Eq.~\eqref{eq:grouped-amplitude-optimal}.
\end{proof}

\subsection{Displacement-action optimality}
\label{app:action-optimality}
This subsection proves the action part of Theorem~\ref{prop:action-optimality}.

\paragraph{Word and explicit bounds.}

Let $h$ be a fixed, real, nonconstant trigonometric polynomial with integer harmonics, let $\bm\nu\ne\bm0$ and $\theta=\bm\nu\cdot\hat{\bm Q}$, and let $0<\eps<1$.  Consider words
\begin{equation}
U(\theta)=R_N\ee^{\ii\alpha_N\sigma_z^{(a)}\theta}R_{N-1}\cdots R_1\ee^{\ii\alpha_1\sigma_z^{(a)}\theta}R_0
\label{eq:action-word}
\end{equation}
with real $\alpha_j$, signal-independent $R_j\in\mathrm{SU}(2)$, and displacement action $A=|\bm\nu|\sum_j|\alpha_j|$.  If
\begin{equation}
\sup_{\theta\in\bbR}\norm{U(\theta)-\operatorname{diag}\bigl(\ee^{-\ii\lambda h(\theta)},\ee^{\ii\lambda h(\theta)}\bigr)}\le\eps
\label{eq:action-criterion}
\end{equation}
for $\lambda>0$, then
\begin{equation}
A\ge|\bm\nu|\left[\lambda\|h'\|_\infty-6\left(\frac{\lambda\|h'''\|_\infty\,\mu_3}{6(1-\eps)}\right)^{1/3}\right],
\label{eq:action-lower}
\end{equation}
where $\mu_3\simeq0.488$ is the third-moment constant of the averaging kernel below.  The grouped GJA word, truncated with the phase-aware tail bound derived below, satisfies Eq.~\eqref{eq:action-criterion} with
\begin{equation}
A=|\bm\nu|K\le|\bm\nu|\left[\lambda\|h'\|_\infty+\calO_h\!\left(\lambda^{1/3}\ell_\eps^{2/3}+\ell_\eps+1\right)\right],
\label{eq:action-upper}
\end{equation}
where $\ell_\eps=\ln[(2+\lambda)/\eps]$.  The minimal displacement action is therefore $\lambda|\bm\nu|\,\|h'\|_\infty+o(\lambda)$ as $\lambda\to\infty$.

Put $\bar\alpha=\sum_j|\alpha_j|$, so that $A=|\bm\nu|\bar\alpha$, and write $f=\ee^{-\ii\lambda h}$.

\paragraph{Frequency support.}
The product expansion in Eq.~\eqref{eq:action-word} expresses every entry of $U(\theta)$ as a finite sum $\sum_\omega c_\omega\ee^{\ii\omega\theta}$.  Each frequency is a signed sum of the $\alpha_j$, so $|\omega|\le\bar\alpha$.  No integer or equal-strength assumption enters.

The lower-bound mechanism is local. We demodulate the target by its largest phase slope. The demodulated target is nearly constant near that point. If the word has insufficient frequency support, a band-limited averaging kernel removes its demodulated response entirely. The approximation error then forces a contradiction.

\paragraph{Averaging kernel.}
With $\operatorname{sinc}u=\sin u/u$, the density $w_b(y)=20b\operatorname{sinc}^6(by)/(11\pi)$ is nonnegative and has unit integral.  Its Fourier transform $\widetilde w_b(\omega)=\int w_b(y)\ee^{\ii\omega y}\,dy$ is proportional to the six-fold convolution of the indicator of $[-b,b]$, so $\widetilde w_b(\omega)=0$ for $|\omega|\ge6b$.  Its third absolute moment is $\int|y|^3w_b(y)\,dy=\mu_3/b^3$, with $\mu_3=5(24\ln2-9\ln3)/(22\pi)\simeq0.48817$.

For completeness, the normalization and moment follow from
\begin{align}
\int_{\mathbb R}\operatorname{sinc}^6u\,du
 &=\frac{\pi}{120}\sum_{k=0}^{3}(-1)^k\binom6k(3-k)^5
 =\frac{11\pi}{20},\\
\int_{\mathbb R}|u|^3\operatorname{sinc}^6u\,du
 &=\frac{24\ln2-9\ln3}{8}.
\end{align}
The first identity is the value at zero of the six-fold interval convolution. The second follows from $\sin^6u=(10-15\cos2u+6\cos4u-\cos6u)/32$, two integrations by parts, and $\int_0^\infty(\cos au-\cos bu)\,du/u=\ln(b/a)$ for $a,b>0$.

\paragraph{Lower bound.}
Choose $\theta_0$ with $|h'(\theta_0)|=\|h'\|_\infty$ and put $d=h'(\theta_0)$.  Then $h''(\theta_0)=0$, and Taylor's theorem gives $|h(\theta_0+y)-h(\theta_0)-dy|\le\|h'''\|_\infty|y|^3/6$.  With $u=U_{00}$ define
\begin{equation}
I_v=\int_{\bbR}w_b(y)\,\ee^{\ii\lambda[h(\theta_0)+dy]}\,v(\theta_0+y)\,dy,\qquad v\in\{u,f\}.
\end{equation}
The demodulated entry $u$ has frequencies $\omega+\lambda d$ with $|\omega|\le\bar\alpha$.  If $\lambda\|h'\|_\infty-\bar\alpha>6b$, all of them lie where $\widetilde w_b$ vanishes, so $I_u=0$.  The Taylor bound and $|\ee^{\ii t}-1|\le|t|$ give $|I_f-1|\le\lambda\|h'''\|_\infty\mu_3/(6b^3)$.  Equation~\eqref{eq:action-criterion} gives $|u-f|\le\eps$ pointwise, hence $|I_f|=|I_f-I_u|\le\eps$.  These estimates are incompatible when $b^3>b_\ast^3=\lambda\|h'''\|_\infty\mu_3/[6(1-\eps)]$.  If $\bar\alpha<\lambda\|h'\|_\infty-6b_\ast$, a choice of $b$ between $b_\ast$ and $(\lambda\|h'\|_\infty-\bar\alpha)/6$ produces this contradiction.  Therefore $\bar\alpha\ge\lambda\|h'\|_\infty-6b_\ast$, which is Eq.~\eqref{eq:action-lower}.  The argument uses the approximation on the whole real coordinate.

\paragraph{Finite-amplitude bound.}
For a bounded almost-periodic function $g$ put
\[
\mathcal M[g]=\lim_{T\to\infty}\frac{1}{2T}\int_{-T}^{T}g(\theta)\,d\theta
\]
and $\widehat g(\omega)=\mathcal M[g\,\ee^{-\ii\omega\theta}]$.  The difference $u-f$ is a finite exponential sum plus the absolutely convergent Fourier series of $f$.  Exponentials with distinct real frequencies have zero mean product, so $\mathcal M[|u-f|^2]=\sum_\omega|\widehat u(\omega)-\widehat f(\omega)|^2$.  By the frequency-support paragraph, $\widehat u(n)=0$ for every integer $n$ with $|n|>\bar\alpha$, and $\widehat f(n)=C_n$.  Hence $\mathcal M[|u-f|^2]\ge\sum_{|n|>\bar\alpha}|C_n|^2$.  Equation~\eqref{eq:action-criterion} gives $|u-f|\le\eps$ pointwise, so $\mathcal M[|u-f|^2]\le\eps^2$.  The integers with $|n|>\bar\alpha$ are those with $|n|>\lfloor\bar\alpha\rfloor$, so $\lfloor\bar\alpha\rfloor\ge K_2(\eps)$ and $A=|\bm\nu|\bar\alpha\ge|\bm\nu|K_2(\eps)$.

\paragraph{Upper bound.}
Write $h(\theta)=h_0+\sum_rb_r\sin(m_r\theta+\phi_r)$ with real $b_r$ and integers $1\le m_r\le L_{\rm h}$.  For real $\theta$ and $s>0$,
\begin{equation}
\begin{split}
\operatorname{Im}h(\theta+\ii s)={}&s\,h'(\theta)+\sum_rb_r\cos(m_r\theta+\phi_r)\\
&\times[\sinh(m_rs)-m_rs].
\end{split}
\end{equation}
Hence $B_h(s)=\sup_\theta|\operatorname{Im}h(\theta+\ii s)|\le s\|h'\|_\infty+c_hs^3$ for $0<s\le1/L_{\rm h}$, with $c_h=\cosh(1)\sum_r|b_r|m_r^3/6$.  A shift of the Fourier contour as in Lemma~\ref{lem:generalized-ja-tail} gives $|C_n|\le\exp[\lambda B_h(s)-|n|s]$ and
\begin{equation}
\tau_K=\sum_{|n|>K}|C_n|\le\frac{2\exp[\lambda B_h(s)-(K+1)s]}{1-\ee^{-s}}.
\label{eq:action-phase-aware-tail}
\end{equation}
For small $s$, the leading growth is therefore set by $\|h'\|_\infty$ rather than by $\sum_r m_r|b_r|$. Cancellation among the harmonics is retained. Since $1-\ee^{-s}\ge s/2$ for $0<s\le1$, the condition $\tau_K\le\eps^2/4$ holds when
\begin{equation}
K+1\ge\lambda\|h'\|_\infty+\lambda c_hs^2+\frac1s\ln\frac{16}{\eps^2s}.
\end{equation}
With $\ell_\eps=\ln[(2+\lambda)/\eps]$, the choice $s=\min\{1/L_{\rm h},(\ell_\eps/\lambda)^{1/3}\}$ bounds the last two terms by $\calO_h(\lambda^{1/3}\ell_\eps^{2/3}+\ell_\eps+1)$.  The rescaling, completion, and factorization of Appendix~\ref{smsec:generalized-ja} then give a paired word with $2K$ balanced queries and paired-unitary error at most $2\sqrt{\tau_K}\le\eps$.  Each balanced query has $|\alpha_j|=1/2$, so $A=|\bm\nu|K$, which is Eq.~\eqref{eq:action-upper}.  The two bounds together give the asymptotic statement.

\section{Circuit implementation}
\label{smsec:circuit-details}

This appendix collects the conventions that turn a compiled word into a circuit. It specifies the finite-energy working space and its error budget (Appendix~\ref{app:working-space}). It also gives the boundary electronic controls (Appendix~\ref{app:boundary-controls}), rotation and carrier conventions (Appendix~\ref{app:processor-conventions}), and rotation-angle recovery (Appendix~\ref{app:angle-recipe}).

\subsection{Representation and finite-energy working space}
\label{app:working-space}

\paragraph{Alternative input representations.}
Mutually commuting Pauli strings can be diagonalized together. Within a diagonal block $d$, write the retained coordinate dependence as
\begin{equation}
f_d(\hat{\bm Q})\approx g_d^{(\le2)}(\hat{\bm Q})+\sum_{j=1}^{R_d}A_{d,j}\cos(\bmu_{d,j}\cdot\hat{\bm Q}+\phi_{d,j})
+ f^{\rm poly}_d(\hat{\bm Q}),
\label{eq:block-phase-model}
\end{equation}
where $g_d^{(\le2)}$ contains the at-most-quadratic Gaussian contribution, $R_d$ counts the retained trigonometric ridges, and $f_d^{\rm poly}$ is the remaining polynomial part.

A fermion-to-qubit mapping of a second-quantized electronic Hamiltonian with coordinate-dependent integrals supplies the same Pauli--coordinate form~\cite{SeeleyRichardLove2012}. In either representation, the present compiler treats dependence on mutually commuting position operators~\cite{HaMacDonell2025,MotlaghVibronic2024,LangMetallic2026,Ollitrault2020NAMD}.

\paragraph{Finite-energy working space.}

The synthesis theorem assumes exact signal access and provides a global operator-norm guarantee.  A physical oscillator realization also needs a declared finite-energy working space.  Let $\bm n^{\max}=(n^{\max}_1,\ldots,n^{\max}_M)$ denote the modewise occupation cutoffs and let $\ket n$ be a Fock number state.  We define
\begin{equation}
\PiN=\bigotimes_{m=1}^{M}\left(\sum_{n=0}^{n^{\max}_m}\ket n\bra n\right).
\label{eq:cutoff-projector}
\end{equation}
Equation~\eqref{eq:cutoff-projector} retains the tensor product of the specified Fock levels.

Let $\hat H_{\rm target}$ be the reference Hamiltonian for evolution time $T$.  The representation stage is required to certify its discrepancy from $\hat H_{\rm target}$ on this working space:
\begin{equation}
\norm{\PiN\!\left(\ee^{-\ii \hat H_{\rm target}T}
-\ee^{-\ii \hat H_{\rm in}T}\right)\!\PiN}\le \eps_{\rm rep}.
\label{eq:representation-certificate}
\end{equation}
The quantity $\eps_{\rm rep}$ collects errors introduced before gate synthesis. These include the chosen electron--nuclear representation and the active-space truncation. They also include any omitted diagonal Born--Huang correction or geometric scalar potential not already contained in $f_\ell$.

Accuracy of the implemented circuit on the same working space is recorded through
\begin{equation}
\norm{\PiN(U-\widetilde U)\PiN}\le \eps_{\rm proj},
\label{eq:projected-norm}
\end{equation}
where $U$ is ideal and $\widetilde U$ is its implementation.  Population that the ideal target itself places outside the cutoff, and additional leakage induced by the approximate circuit, are tracked separately. The two-mode benchmarks test specific joint input states under finite-matrix signals and do not certify this entire representation contract or full-space leakage for a hardware displacement.

The compiler assumes a structured description of the coordinate dependence. Qubit-only and Gaussian terms use native gates, and polynomial phases may use existing analytic OQ--GQSP methods under their stated input and postselection conditions~\cite{HongOQGQSP2026,Yang2026RabiPolyPhase}. A black-box potential-energy surface must first be replaced by a structured approximation, such as a sinusoidal-network fit~\cite{AertsSinNN2026}, whose certified residual contributes to $\eps_{\rm rep}$.

\paragraph{Hamiltonian and unitary Fourier approximation.}
Section~\ref{sec:hamiltonian-first} defines the Hamiltonian-first and direct (unitary-first) routes and their separated error bound. For the same supplied $V_R$ and periodic extension, exact GJA coefficients and exact unitary Fourier coefficients coincide, so an FFT may also evaluate the GJA coefficients. A fair comparison also matches the circuit architecture. Partially serial GJA and direct Fourier synthesis can share the same commuting groups and one-variable completion, whereas a global multivariable FFT or multi-index GJA does not by itself supply a native circuit (Appendix~\ref{app:multidirection-response}). For nonperiodic inputs, the coordinate window and periodic extension must also be specified, as described next.

\paragraph{Nonperiodic coordinate windows.}\label{app:coordinate-window}
For a physical interval $\mathcal I_Q=[L_{\min},L_{\max}]$, set $Q_c=(L_{\min}+L_{\max})/2$ and $w_Q=(L_{\max}-L_{\min})/2$. A Fourier extension uses a larger period $2\zeta_Q w_Q$, $\zeta_Q>1$, with signal coordinate $\theta=\xi_{\rm win}(Q-Q_c)$ and $\xi_{\rm win}=\pi/(\zeta_Q w_Q)$~\cite{AdcockHuybrechsMartinVaquero2014,MatthysenHuybrechs2018}. A translation of the same profile changes only signal-independent rotations. Moving the window across a fixed potential changes the target being approximated. A degree-$K$ balanced word has $2K$ balanced queries and CD action $A=K|\xi_{\rm win}|$. Larger values of $w_Q$ or $\zeta_Q$ reduce the kick size but can increase the required degree. They therefore give no monotone action saving.

Accuracy on $\mathcal I_Q$ does not imply the whole-circle contractivity that completion requires, and any repair is charged to the paired-unitary error. A real Hamiltonian extension keeps unit modulus globally but can raise the response degree off-window, so neither route has a universal cost ordering. The benchmarks in Sec.~\ref{sec:benchmarks} use supplied trigonometric targets and do not optimize the window or extension. Coordinate-window errors do not certify the whole real line, and a Fock cutoff alone does not bound window tails.

\paragraph{Error allocation.}

Coherent errors add by unitary telescoping when the ideal and implemented gates preserve a common working subspace and their operator-norm errors are bounded there. The contributions are representation, outer evolution, ridge synthesis, phase recovery, signal implementation, and boundary controls. If ridge word $b$ contains $N_{W,b}$ queries with error at most $\eta_{W,b}$ per call on that invariant space, its signal contribution is at most $\sum_b N_{W,b}\eta_{W,b}$. If the initial and final boundary controls have operator-norm errors $\eta_{B,\rm in}$ and $\eta_{B,\rm out}$ on the same invariant space, they contribute at most $\eta_{B,\rm in}+\eta_{B,\rm out}$ in addition to the scalar-word error. Two-sided projected errors such as Eq.~\eqref{eq:projected-norm} alone do not support this composition bound when intermediate states can leave the subspace. In that case, bounds must also control the output components and leakage on the spaces reached during the sequence. Channel noise and statistical error retain their own metrics. Derivative couplings require their own access model.

\subsection{Boundary electronic control}
\label{app:boundary-controls}

The boundary construction of Sec.~\ref{sec:controlled-lift} follows from the ancilla flip. Since $\sigma_x\sigma_z\sigma_x=-\sigma_z$ and $B_P$ commutes with $g(\bQ)$, conjugation gives the first identity in Eq.~\eqref{eq:boundary-pauli-error}. Unitary invariance of the norm gives the second.

Figure~\ref{fig:pauli-lift-compilers} shows the boundary lift for serial single-tone JA and grouped GJA synthesis, and the same lift applies to the TGIFS cells of Eqs.~\eqref{eq:tgifs-gqsp-cell} and \eqref{eq:tgifs-product-hierarchy}.  In fact, $B_P$ commutes with every $X$-basis CD and conjugates $R_z(\varphi)$ to $\exp[-\ii\varphi\sigma_z^{(a)}P/2]$.  For the first-order and second-order TGIFS words, moving the two boundary gates inward therefore reproduces the original electronically conditioned rotations exactly.  Conversely, they can be moved out of the repeated cells without changing the ideal word.  A Direct Fourier circuit with the same paired-unitary guarantee can use the same boundary lift.  Consecutive words with the same $P$ share one compute/uncompute pair, as in the serial circuit of Fig.~\ref{fig:pauli-lift-compilers}(a). Boundary-control savings are therefore not exclusive to grouped synthesis.

\begin{figure*}[!t]
\centering
\includegraphics[width=0.98\textwidth]{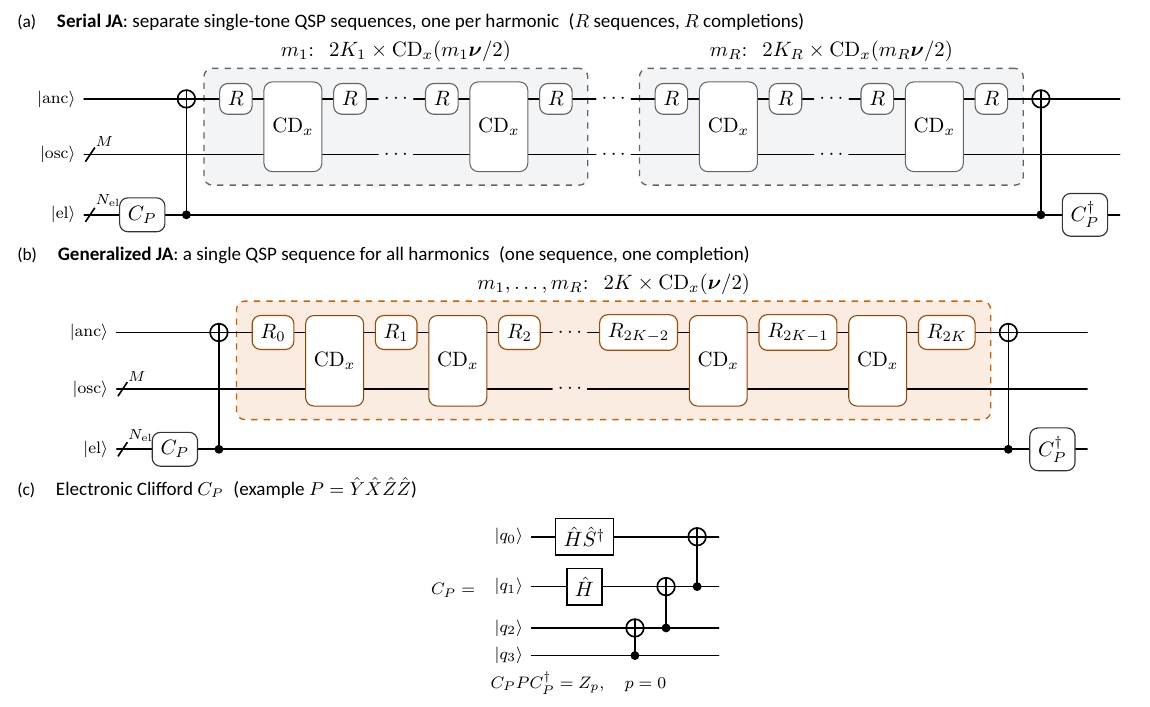}
\caption{Boundary electronic control for serial and grouped synthesis. Time runs left to right, starting at $R_0$. When a circuit diagram bundles a register into one rail, $M$ labels oscillator modes and $\Nel$ labels electronic qubits.  Neither is the Fock-cutoff vector $\bm n^{\max}$. (a) Serial JA uses one completed word per harmonic, with $2K_r$ queries of weight $m_r$. (b) Grouped GJA uses $2K$ balanced queries. Both share two pivot-controlled ancilla flips and one Clifford compute/uncompute pair, independent of query count. Bundled displacements commute but need not run simultaneously. (c) Example Clifford mapping $P=YXZZ$ to $Z_0$. Pure tones admit phase-only rotations. Generic completions can require active rotations. Endpoint basis changes, rotation decompositions, and optional query carriers are specified in Appendix~\ref{app:processor-conventions}.}
\label{fig:pauli-lift-compilers}
\end{figure*}

\subsection{Alternative joint-signal access}
\label{app:joint-signal-access}
The boundary construction suffices for the circuits in Figs.~\ref{fig:schematic} and \ref{fig:pauli-lift-compilers} and for grouped synthesis. An alternative formulation specifies the joint-signal access used by the single-tone theorem and its matching lower bound.  It is not an additional step in the boundary circuit.  For that formulation, electronic control is included in the signal itself:

\begin{equation}
\mathcal W_{P,\bmu,\phi}=P\otimes W_{\bmu,\phi}
\label{eq:joint-pauli-signal}
\end{equation}
Its Hermitian and anti-Hermitian parts are $P\otimes\cos(\bmu\cdot\hat{\bm Q}+\phi)$ and $P\otimes\sin(\bmu\cdot\hat{\bm Q}+\phi)$.

The balanced query applies $P\ee^{\ii\hat x/2}$ on the upper ancilla branch and $\ee^{-\ii\hat x/2}$ on the lower, with $\hat x=\bmu\cdot\hat{\bm Q}+\phi$.  Their relative operation is $\mathcal W_{P,\bmu,\phi}$. Only the oscillator exponent is split between the branches. No fractional power of $P$ appears.  Hadamard conjugation of the ancilla gives

\begin{equation}
\bigl[H_a\,\mathrm{ctrl}_0(P)\,H_a\bigr]
\ee^{\ii\phi\sigma_x^{(a)}/2}\CDx{\bmu/2}.
\label{eq:balanced-x-query}
\end{equation}

Each query now carries an electronic parity control, which becomes $\mathrm{ctrl}_0(\sigma_z^{(p)})$ in the pivot basis. A $2K$-query word carries the known frame $P^K$, which is removed once at the boundary, and the two realizations have separate electronic-control costs. Direct substitution of $P\otimes W_{\bm\nu}$ into the polynomial keeps $P$ only on odd powers, because $(P\otimes W)^m=P^m\otimes W^m$, whereas the boundary lift has no such parity restriction.

\subsection{Rotation and carrier conventions}
\label{app:processor-conventions}

Arbitrary ancilla rotations are not required for every coordinate-phase target.  A supplied term can be written as $a\sin(x+\phi)=a\cos\phi\sin x+a\sin\phi\cos x$ and treated by separate single-tone JA constructions.  For a common Pauli string, these coordinate-only generators commute, so splitting their ideal exponentials introduces no product-formula error.  

Endpoint-admissible phase-only words retain virtual-$Z$ interior rotations.  This structure remains after grouping alone.  A word whose queries are $\mathrm{CD}_x$ gates and whose processing rotations are all virtual-$Z$ gates obeys $U(-\theta)=\sigma_z^{(a)}U(\theta)\sigma_z^{(a)}$.  It can therefore meet the paired target only if $|\ee^{\ii h(-\theta)}-\ee^{\ii h(\theta)}|\le2\eps$ for all $\theta$.  Other mixtures need fixed boundary rotations, a split into even and odd parts, or generic rotations.

\paragraph{Virtual-$Z$ implementation of rotations.}
Under phase-tracked control, a generic rotation admits the Euler form
\begin{equation}
R_j=R_z(a_j)R_x(b_j)R_z(c_j).
\label{eq:processor-zxz}
\end{equation}
The $Z$ factors are virtual. One active transverse pulse remains if arbitrary rotation angles are supported.  Alternatively, writing $X_{90}=R_x(\pi/2)$ gives the exact fixed-pulse decomposition~\cite{McKay2017VirtualZ}
\begin{align}
R_j={}&R_z(a_j-\pi/2)X_{90}R_z(\pi-b_j)\nonumber\\
     &\times X_{90}R_z(c_j-\pi/2).
\label{eq:processor-fixed-x90}
\end{align}
All variable angles then reside in virtual-$Z$ updates, but two active $X_{90}$ pulses remain before any exact cancellations, so this is not a phase-only QSP word. These rewrites preserve the ideal word and its CD time, not a calibrated device duration, and the frame must be tracked through later CD and boundary operations. Active-rotation duration is a separate hardware-dependent cost and is not optimized in the reported scans.

\paragraph{Offset and carrier conventions.}

The shifted-sine proof of Appendix~\ref{smsec:ja-details} uses the endpoint repair of Ref.~\cite{LowChuang2017}, whose inverse factorization needs only phase rotations~\cite{LowYoderChuang2016,Haah2019QSP}, and places a fixed $\pi/2$ rotation about $X$ in every query. The benchmark words instead recover the cosine response directly and need no such carrier (Appendix~\ref{app:elementary-duration}). A nonzero offset adds the carrier $R_x(-\phi)$ of Eq.~\eqref{eq:qsp-execution-template} to every query in all three families. The carrier commutes with $\mathrm{CD}_x$, its time is proportional to $T_{\rm CD}$, and an offset in the compiled response does not make it a virtual-$Z$ gate.

\subsection{Angle recovery}
\label{app:angle-recipe}

The compiler takes the signal vector $\bm\nu$, the harmonic amplitudes and offsets, and a target error $\eps$. It returns the number of conditional displacements and the ancilla rotations of Eq.~\eqref{eq:qsp-execution-template}. The workflow has four steps and a final check.

\paragraph{Coefficients.}
The Fourier coefficients $C_n$ of $G=\exp(-\ii h)$ are computed by convolving the single-tone Bessel series and are cross-checked against a fast Fourier transform of $G$.

\paragraph{Truncation.}
Choose $K$ with a tail bound $\tau_K$ and form $F_K=T_K/(1+\tau_K)$ as in Eq.~\eqref{eq:generalized-ja-contractive}. In exact arithmetic, $\tau_K\le\eps^2/4$ suffices for the paired-unitary bound. Numerical coefficient, completion, and recovery errors need additional allowances in the same budget.

\paragraph{Completion.}
The complementary polynomial $H_K$ satisfies $|F_K|^2+|H_K|^2=1$ on the unit circle. We obtain it from the roots of $1-|F_K|^2$ and keep the roots inside the unit circle. A different choice of roots gives a different valid circuit, so the choice is part of the compiler output.

\paragraph{Rotations.}
The completed $2\times2$ matrix is factored into $2K$ balanced queries and the rotations $R_0,\ldots,R_{2K}$ by removing one query at a time~\cite{Haah2019QSP}. This factorization is an inverse nonlinear Fourier transform (NLFT)~\cite{AlexisNLFA2024,LaneveNLFT2025}. Dedicated inverse-NLFT algorithms, such as Riemann--Hilbert--Weiss factorization~\cite{LaneveNLFT2025} and the inverse nonlinear fast Fourier transform~\cite{NiInverseNLFFT2025}, solve the same problem and can replace the completion and factorization steps. The response contains both positive and negative Fourier modes, as in generalized QSP~\cite{MotlaghWiebe2024,YamamotoYoshioka2024Prony}. The angles reported in this work use the root-based completion and sequential factorization described here. Each rotation is then written in the ZXZ form of Eq.~\eqref{eq:processor-zxz} or the fixed-$X_{90}$ form of Eq.~\eqref{eq:processor-fixed-x90}.

\paragraph{Final check.}
We rebuild the complete $2\times2$ word from the emitted angles, including the boundary rotations and the scalar phase, and compare it with $\operatorname{diag}(G,G^*)$ over the whole signal circle. The comparison uses a dense grid with a derivative bound between grid points. A word is accepted only if this paired-unitary error stays within its allowance, and otherwise the precision or the truncation is revised.

\paragraph{Single-tone phase-only word.}
\label{app:elementary-duration}
For the single-tone benchmark of Fig.~\ref{fig:compiler-benchmarks}(a), the rotations can be restricted to the equatorial plane~\cite{Haah2019QSP}. At $L_Q=10$ the word is
\begin{align}
U_K^{\rm QSP}(\theta_{L_Q})
&=R_z(\varphi_{\rm res})\prod_{j=1}^{2K}E_{\beta_j}(\theta_{L_Q}),\nonumber\\
E_{\beta}(\theta_{L_Q})
&=\exp\!\left[\frac{\ii\theta_{L_Q}}{2}
(\cos\beta\,\sigma_x+\sin\beta\,\sigma_y)\right],\nonumber\\
\theta_{L_Q}&=\frac{\pi}{L_Q}\hat Q+\phi_0.
\label{eq:elementary-phase-only-qsp}
\end{align}
Virtual-$Z$ frames program the equatorial axes $\beta_j$ and the residual $R_z(\varphi_{\rm res})$. The word therefore uses $2K$ balanced queries and $2K+1$ virtual parameters. It needs neither an active rotation pulse between the queries nor the fixed $\pi/2$ query carrier of the shifted-sine construction. A nonzero $\phi_0$ still requires its separately counted offset carrier. Each selected word is checked on a shifted $2^{20}$-point grid with derivative padding $(K+|\Lambda|)\pi/2^{20}$. Independent native-gate replays verify the target, both ancilla branches, and the CD action. Appendix~\ref{app:common-metric} gives the selection criterion.

\subsection{TGIFS cells and sufficient error bounds}
\label{smsec:tgifs-details}\label{lem:tgifs-gqsp-embedding}\label{prop:tgifs-block-metrics}
Here and below in the scalar-cell formulas, $\sigma_x$, $\sigma_y$, and $\sigma_z$ act on the synthesis ancilla. For the zero-offset cosine, the exact first-order cell is
\begin{align}
\mathcal G^{\rm basic}_{P}(\vartheta)
&=\CDx{\bmu/2}\ee^{\ii\vartheta\sigma_z^{(a)}P}
  \CDx{-\bmu}\nonumber\\
&\quad{}\times\ee^{\ii\vartheta\sigma_z^{(a)}P}
  \CDx{\bmu/2}.
\label{eq:tgifs-gqsp-cell}
\end{align}

CD conjugation of $\sigma_z^{(a)}P$ gives generators $A_\pm=(\sigma_z\cos x\pm\sigma_y\sin x)P$ with
\begin{equation}
[A_+,A_-]=2\ii\sigma_x\sin(2x).
\label{eq:tgifs-generator-commutator}
\end{equation}
The two-exponential norm estimate bounds one cell's error relative to $e^{2\ii\vartheta\sigma_zP\cos x}$ by $\vartheta^2$. Thus for $\vartheta_r=-\Lambda/(2r)$, unitary telescoping gives
\begin{equation}
\|[\mathcal G_P^{\rm basic}(\vartheta_r)]^r-e^{-\ii\Lambda\sigma_zP\cos x}\|
\le\min\{2,\Lambda^2/(4r)\}.
\label{eq:small-angle-lt-error}
\end{equation}
This bound is sufficient and makes no claim of a matching lower bound for optimized TGIFS families. The exact block entries of the cell are given in Ref.~\cite{Rainaldi2026TrigBenchmark}. Its selected block errs at $\calO(\vartheta^3)$, whereas the amplitude transferred to $\ket1_a$ is $\calO(\vartheta^2)$, so selected-block and paired-unitary errors must not be interchanged.

Since $(\sigma_zP)\mathrm{CD}_x(\bmu)(\sigma_zP)=\mathrm{CD}_x(-\bmu)$ and $e^{\ii\pi\sigma_zP/2}=\ii\sigma_zP$, substitution in the middle backward displacement yields
\begin{equation}
\begin{split}
[\mathcal G_P^{\rm basic}(\vartheta)]^r
&=(-1)^r\Bigl[\CDx{\bmu/2}\\
&\quad{}\times
 e^{\ii(\vartheta+\pi/2)\sigma_zP}\CDx{\bmu/2}\Bigr]^{2r}.
\end{split}
\label{eq:tgifs-forward-word}
\end{equation}
The word has $4r$ forward balanced queries and $2r$ equal rotations. The sine cell is obtained by the coordinate substitution $x\mapsto\pi/2-x$ together with the ancilla-axis conjugation $(\sigma_y+\sigma_z)/\sqrt2$.

\subsection{Symmetric TGIFS schedule and carrier cost}
\label{app:tgifs-second-order}
We also construct a symmetric second-order TGIFS cell as an extension of the first-order reference. Written in the native primitive of Eq.~\eqref{eq:tgifs-elementary-qsp}, it is
\begin{align}
\mathcal G^{\rm sym}_{P}(\vartheta)
&=\CDx{\bmu/2}\ee^{\ii\vartheta\sigma_z^{(a)}P/2}
  \CDx{-\bmu}\nonumber\\
&\quad{}\times\ee^{\ii\vartheta\sigma_z^{(a)}P}
  \CDx{\bmu}\nonumber\\
&\quad{}\times\ee^{\ii\vartheta\sigma_z^{(a)}P/2}
  \CDx{-\bmu/2}.
\label{eq:tgifs-product-hierarchy}
\end{align}
A multi-harmonic ridge sum concatenates the corresponding cells. The outer Fourier truncation and each cell's block error are accounted for separately. For harmonic $m\ge1$ with Fourier coefficient $a_m$, the four signal units
in one unmerged second-order TGIFS cell have chronological signed weights
$(-m,+2m,-2m,+m)$.  A unit of signed weight $s$ applies
\begin{equation}
\CDx{s\bmu/2}\,\Rx(-s\phi_0)
=\exp\!\left[\frac{\ii s}{2}\sigma_x
(\bmu\cdot\bQ+\phi_0)\right],
\end{equation}
so its $\sigma_x=\pm1$ branches acquire the opposite phases
$\exp[\pm\ii s(\bmu\cdot\bQ+\phi_0)/2]$.  At $P=I$ and $\vartheta_r=-a_m/(2r)$, this cell is repeated $r$ times.

\paragraph{Boundary cancellation.}
At every repeat boundary the closing $(+m)$ signal unit cancels the next opening $(-m)$ unit exactly.  The two adjacent $R_z(a_m/2r)$ rotations merge into $R_z(a_m/r)$.  The merged schedule therefore begins and ends with a single signal unit.  All interior signal units are double units.  It has two end rotations at $a_m/(2r)$ and $2r-1$ interior rotations at $a_m/r$. Equation~\eqref{eq:tgifs-product-hierarchy} specifies the exact unmerged cell. Boundary-canceled pairs are neither executed nor counted in the compiled schedule.
\paragraph{Action and carrier time.}
Its weighted CD action is $6mr-2m(r-1)=m(4r+2)$, against
$4mr$ for first-order TGIFS.  In the accounting, the conditional displacement contributes $\abs{s}\tau_{\rm base}$ of CD time.  The equatorial remainder contributes a rotation of absolute angle $\abs{s\phi_0}$.  It is measured in units of a $\pi/2$ rotation about $X$ (X90).  The virtual $\Rz$ frames carry zero cost. The rotation action is $\abs{\phi_0}(T_{\rm CD}/\tau_{\rm base})/(\pi/2)$, where $T_{\rm CD}/\tau_{\rm base}$ sums the contributions $m(4r+2)$ over the harmonics, so it is exactly proportional to the CD time.  In contrast, the processing angles of the Jacobi--Anger method are outputs of the synthesis and are implemented as virtual-$Z$ frame updates in the declared gate model.

\paragraph{Error curves in Fig.~\ref{fig:tgifs-ja-nlft-dictionary}(c).}\label{app:fig2-bound}
The curves compare paired-unitary errors at $\Lambda=2$ under ideal exact signal access. The TGIFS curves show the errors of the exact first-order ($T_{\rm CD}=4r\tau_{\rm base}$) and second-order symmetric ($T_{\rm CD}=(4r+2)\tau_{\rm base}$) words. The errors are maximized over 2001 uniformly spaced signal phases. These sampled maxima are not certified continuous suprema. The JA curve shows the completed-word upper bound $4\sqrt{\delta_K}$ at $T_{\rm CD}=2K\tau_{\rm base}$. It uses the repair step and Lemma~\ref{cor:ja-full-joint-completion}. Here $\delta_K=2\exp[1/(K+2)]/\{(K+1)![1-1/(K+2)]\}$ bounds the uniform Jacobi--Anger truncation tail at this amplitude. This majorant follows from $|J_n(2)|\le e^{1/(n+1)}/n!$ and a geometric bound on the remaining factorial series. The JA curve is shown where $8\delta_K<1$ and continues below the displayed error range.

\section{Numerical validation}
\label{app:numerical-validation}

\paragraph{Scope of the numerical checks.}
Each benchmark uses its own admission criterion, specified in the subsections below. Full-circle paired-unitary checks use dense grids with analytic derivative padding, and circuit-benchmark words are also replayed by an independent native-gate implementation. The sinusoidal-route comparison is validated before native-angle recovery. As in Sec.~\ref{sec:benchmarks}, these checks are floating-point evidence and bounded-search upper bounds, and the reported orderings depend on the target, criterion, tolerance, and reference CD clock.

\subsection{Coordinate-domain error allocation}
\label{app:fig3-interval-budget}

\paragraph{Distinct-group action comparison.}
\label{app:fig3-serial-action}
Figure~\ref{fig:structure-programming}(b) uses $V(q)=(1-\ee^{-q})^2$ and the domain $I^2$ with $I=[-\log(1+\sqrt{1/2}),-\log(1-\sqrt{1/2})]$. The three coordinates are $q_1$, $q_2$, and $(q_1+q_2)/2$, with phase weights $2$, $1.5$, and $0.8$ per three-block sequence. A common real $C^3$ periodic extension equals $V$ on $I$ and closes the remaining interval with a degree-seven Hermite polynomial matching derivatives through order three. HF truncates the potential at degrees $R_g$ before exponentiation and response truncation at $K_g$. UF directly truncates the exponential of the common extension at $K_g$. Both use global contractivity repair and the same inside-root completion with central coefficient positive real. No convex response or completion-gauge optimization is used.

Let $W_j(\bm q)$ be the completed product through block $j$, and let $D_j(\bm q)$ be the exact paired target for the same prefix. The plotted metric is the sampled paired-unitary error
\begin{equation}
 \max_{1\le j\le48}\max_{\bm q\in\mathcal G_{513}^2}
 \|W_j(\bm q)-D_j(\bm q)\|_{\rm op},
\end{equation}
where $\mathcal G_{513}$ is a common 513-point grid on $I$. The full ancilla matrices are multiplied coherently, including the complementary branches. This sampled maximum is not a continuum supremum certificate.

The displacement directions are $\bm\nu_1=(1,0)$, $\bm\nu_2=(0,1)$, and $\bm\nu_3=(1/2,1/2)$. Every direction has unit $\ell^1$ norm, so each balanced query costs the same $\tau_{\rm base}$ and $T_{\rm CD}/\tau_{\rm base}=32\sum_gK_g$. The plotted HF candidate with $R=(10,9,9)$ and $K=(33,27,21)$ has $T_{\rm CD}=2592\tau_{\rm base}$ and sampled error $0.0973455$. The UF candidate with $K=(42,46,36)$ has $T_{\rm CD}=3968\tau_{\rm base}$ and sampled error $0.0993176$. Vector-query counts, single-mode pulse counts, and calibrated duration are distinct resources. Panel (b) evaluates completed responses; it does not report native-angle or noisy-circuit replay of these frontier candidates.

\subsection{Matched sinusoidal construction and serial composition}
\label{app:sinusoidal-route-comparison}

As a matched-model control, both routes receive $v_R(\theta)=\cos\theta+0.12\cos2\theta$ at $\Delta t=2$. Bessel-sum GJA coefficients ($|m|\le96$) and a $65536$-point FFT of $\exp[-2\ii v_R]$ agree to $2.7\times10^{-16}$ for $|n|\le128$, and the contour bound on $|\operatorname{Im}\theta|=1$ limits the remaining tail to $7.5\times10^{-55}$. With the same truncation $K=6$, rescaling, and inside-root completion, both routes give sampled paired error $0.0973342$ and leakage $0.0094512$. Serial two-mode layers $W(Q_1+Q_2)W(Q_2)W(Q_1)$ built from the two routes differ by less than $5\times10^{-14}$ in vacuum-averaged leakage and joint-state error for up to $32$ layers. No native-angle recovery, oscillator truncation, kinetic step, or noise model is included.

\subsection{Full-grid common-metric comparison}
\label{app:common-metric}

\paragraph{Criterion and recovery}
For each of the sixteen amplitudes in Fig.~\ref{fig:compiler-benchmarks}(a), we select the shortest circuit of each family whose paired-unitary error is below $10^{-2}$ on the whole signal circle. JA candidates come from a $129$-point full-circle least-squares objective with every lower degree scanned, and TGIFS uses basic and symmetric cells with upward repetition scans.

\paragraph{Independent checks}
Each selected word is checked on a shifted $2^{20}$-point grid with derivative padding. The derivative bound is $K+|\Lambda|$ for a JA word. For TGIFS, differentiating the exact cell factors gives $2r|\sin t|$ for the basic cell and $r[2|\sin(t/2)|+|\sin t|]$ for the symmetric cell, where $t=\Lambda/(2r)$, and adding $|\Lambda|$ accounts for the target derivative. A separate dense native-gate implementation replays the selected words at $N_F=32$ and $40$ with matrix disagreements below $5.0\times10^{-13}$.

\paragraph{Metric dependence}
An entanglement-infidelity threshold of $10^{-4}$ restricted to low-energy inputs, the oscillator levels $n\le10$ with both ancilla states, instead moves the first strict JA reduction from $\Lambda=0.6$ to $0.7$ and changes the counts at $\Lambda=5$ to $(20,448,54)$. The onset of a sampled saving therefore depends on the error metric and tolerance.

\subsection{Multiharmonic compilation and two-mode benchmark details}
\label{app:grouped-benchmarks}

\paragraph{Dense eight-harmonic profile and validation}
\label{app:dense-eight}
The profile is given in Eq.~\eqref{eq:dense-eight-profile}, with $C=1.9866538013267567$ and $\sum_m|a_m|=3.03446$. GJA and independent Direct Fourier recovery use contractive repair, root completion, and native factorization with degrees $K\le96$, and they give identical grouped CD costs at all eleven scales. Figure~\ref{fig:compiler-benchmarks}(b) displays grouped GJA.

\emph{Search scope.} Serial JA uses a phase-only single-tone library and local degree and order allocation. TGIFS uses basic and symmetric cells with repetition allocation, mirror choices, and adjacent-CD merging.

\emph{Independent validation.} All accepted words are checked on $2^{20}$ full-circle phases with analytic first-derivative padding. A bound $M_1$ on the derivative of the recovered word minus its target extends the grid maximum by $M_1\pi/2^{20}$. A second implementation replays stored native rotations or literal TGIFS cells on a shifted $2^{20}$-point grid and independently recounts CD action. For grouped and serial words, respectively, the derivative bounds $M_1$ are $K+s\sum_m m|a_m|$ and $\sum_m mK_m+s\sum_m m|a_m|$. The TGIFS bound adds the target derivative bound to the sum of retained tone derivatives. Numerical guards accompany both grids.

\paragraph{Two-mode benchmark target}

The cubic Fermi normal form $Q_sQ_b^2$ couples a stretch excitation to two
bend quanta~\cite{Chalermpusitarak2025TGIFS}. A bare cubic added to harmonic
confinement is unbounded below. We therefore declare the bounded target
\begin{align}
 F_{d_{\rm reg}}(q_s,q_b)&=\frac{2\sin(d_{\rm reg}q_s)[1-\cos(d_{\rm reg}q_b)]}{d_{\rm reg}^3},\\
 R_{d_{\rm reg}}(q_s,q_b)&=\frac{4F_{d_{\rm reg}}(q_s,q_b)-F_{2d_{\rm reg}}(q_s,q_b)}{3},
 \label{eq:fermi-regularization}
\end{align}
Here $d_{\rm reg}=0.25$ is the regularization spacing. The identity
$2\sin u\cos v=\sin(u+v)+\sin(u-v)$ supplies the three-ridge form.
Taylor expansion gives $R_{d_{\rm reg}}=q_sq_b^2+O(d_{\rm reg}^4)$ on compact coordinate sets.
Since $R_{d_{\rm reg}}$ is bounded, the Hamiltonian in Eq.~\eqref{eq:fermi-H} is bounded below.
The regularization itself, rather than an unspecified full-space cubic
approximation, is the common benchmark target.
We set the harmonic zero-point energy to zero and use
$U(t)=\exp[-\ii\,2\pi cH_{\rm cm^{-1}}t]$, with
$c=2.99792458\times10^{-5}$~cm/fs.

\paragraph{Effective parameters and detuning control.}

The carbon-dioxide dyad values $|W_{\rm F}|=51.232$~cm$^{-1}$ and
$\delta=-7.87$~cm$^{-1}$ are quoted in
Ref.~\cite{McCluskeyStoker2006CO2}. We use their magnitude and detuning,
not a fitted potential-energy surface. The physical carbon-dioxide bend is degenerate, and the effective two-mode model does not represent the full molecular Hamiltonian. The two reported dyad line positions
give the effective mean $1336.770$~cm$^{-1}$. A two-level parameterization uses the mean $\bar\omega=(\omega_s+2\omega_b)/2$ and detuning $\delta=\omega_s-2\omega_b$, so $\omega_s=\bar\omega+\delta/2$ and $\omega_b=(\bar\omega-\delta/2)/2$. This sets $\omega_s=1332.835$~cm$^{-1}$ and
$\omega_b=670.3525$~cm$^{-1}$.
These are inferred effective parameters, not measured uncoupled fundamentals.
The infinite-Fock doorway element is
\begin{equation}
 M_{d_{\rm reg}}=\bra{0,2}R_{d_{\rm reg}}\ket{1,0}
 =\frac{2e^{-d_{\rm reg}^2/2}-\tfrac12e^{-2d_{\rm reg}^2}}{3}.
\end{equation}
Numerically $M_{d_{\rm reg}}\simeq0.499073$, so
$\kappa=|W_{\rm F}|/M_{d_{\rm reg}}\simeq102.654$~cm$^{-1}$.
An artificial off-resonant control changes only the bend frequency to $794.4975$~cm$^{-1}$. This gives $\delta=-5|W_{\rm F}|$.

\paragraph{Ridge representation.}

Equation~\eqref{eq:fermi-regularization} has three base vectors,
$d_{\rm reg}(1,1)$, $d_{\rm reg}(1,-1)$, and $d_{\rm reg}(1,0)$.
Each diagonal direction has sine coefficients
\begin{equation}
 (b_1,b_2)=\left(-\frac{4\kappa}{3d_{\rm reg}^3},\frac{\kappa}{24d_{\rm reg}^3}\right),
\end{equation}
and the central direction has $-2(b_1,b_2)$. A step $\Delta t$ uses
$a_m=2\pi c\Delta t\,b_m$ in the scalar response
$\exp[-\ii\sum_{m=1}^2a_m\sin(m\theta)]$.
Only two distinct schedules must be recovered, but three physical ridge
words are applied. Distinct directions are not combined into a single
multivariate QSP word. All electronic labels are $P=I$.

\paragraph{Input, finite-basis signal, and evolution step.}
We diagonalize Eq.~\eqref{eq:fermi-H} with $N_F=24$ states per mode
(dimension $576$), starting from $\ket{1,0}$.
Every compiled trajectory uses the same Strang step $e^{-\ii H_0\Delta t/2}\widetilde U_Ve^{-\ii H_0\Delta t/2}$. Here $H_0$ is the harmonic part of the benchmark Hamiltonian.
We use $\Delta t=0.25$~fs over $1$~ps. The step-selection and cutoff checks are reported below.
The numerical signal is exactly $e^{\ii\bm v\cdot\bm Q_{N_F} X}$ in the
finite $Q_{N_F}$ eigensystem. The finite-matrix signal defines this benchmark. Physical displacement calibration and full-space leakage require separate checks.

\paragraph{Joint-state error.}
We multiply the full $2\times2$ ridge words in the order $(+,-,s)$
at each coordinate point and propagate both ancilla components through
all $4000$ steps. No component is discarded, measured, reset, or
renormalized. The reported state distance, which bounds the joint trace distance $D$, is
\begin{equation}
 e_{\rm joint}=\max_{t\in\{0,2,\ldots,1000\}\,\mathrm{fs}}
 \left\|\ket{\Psi_{\rm comp}(t)}
       -\ket0\otimes\ket{\psi_{\rm ref}(t)}\right\|_2 ,
\end{equation}
without removing a global phase. Population measurements sum both
ancilla outcomes. The ancilla leakage, compiler-only error against
exact-$V$ Strang propagation, and norm residual are recorded separately.

\paragraph{Budgets and candidate searches.}
The six common per-$V$ error budgets are $\eps_V=10^{-3},3\times10^{-4},10^{-4},3\times10^{-5},10^{-5},3\times10^{-6}$. Each method allocates $\eps_V/3$ to each of the three ridges. This fixed allocation is not optimized over ridge partitions. TGIFS searches both two-tone orders and relative mirror choices. Admission checks each actual completed word. Independent Direct Fourier recovery gives the same CD costs and nearly identical joint-state errors as GJA across these budgets, including the least-cost passing candidate at $52\,000\tau_{\rm unit}$.

\paragraph{Full-circle admission.}
For each two-mode ridge profile, let $E(\theta)$ be the complete recovered word minus its paired target.  Full-circle admission uses a $2^{16}$ grid
and second-derivative interpolation. If $\norm{E''}\le M_2$, then
\begin{equation}
 \sup_\theta\norm{E(\theta)}
 \le \max_{\rm grid}\norm{E(\theta)}
       +\frac{M_2}{8}\left(\frac{2\pi}{2^{16}}\right)^2 .
\end{equation}
This follows by linear interpolation of the matrix-valued error and
convexity of the norm, including the periodic seam.
For $2K$ unitary balanced queries,
$M_2=K^2+S_1^2+S_2$ suffices, where
$S_j=\sum_m m^j|a_m|$ for $j=1,2$.
The TGIFS bound is at most $2(S_1^2+S_2)$, obtained by differentiating
its unitary factors. It is independent of the repetition count.
A $2^{20}$ replay checks representative grouped cases.

\paragraph{Per-step CD cost.}
With the declared fixed-drive serial-mode reference CD clock,
$\tau_{\rm CD}(\bm v)=\tau_{\rm unit}\norm{\bm v}_1$,
the grouped per-step cost is
\begin{equation}
 T_{\rm CD}^{(V)}
 =d_{\rm reg}(4K_{\rm diag}+K_s)\tau_{\rm unit}.
\end{equation}
Each diagonal ridge $d_{\rm reg}(1,\pm1)$ uses $2K_{\rm diag}$ queries of $\ell^1$ action $d_{\rm reg}$. Together they contribute $4d_{\rm reg}K_{\rm diag}\tau_{\rm unit}$. The central ridge $d_{\rm reg}(1,0)$ uses $2K_s$ balanced queries of action $d_{\rm reg}/2$. It contributes $d_{\rm reg}K_s\tau_{\rm unit}$.
TGIFS costs use the actual signed displacement vectors after adjacent
signal merging, within and between ridge words. Sine-carrier $R_x$
rotations and virtual-$Z$ updates are counted separately. A higher
harmonic is charged its longer displacement time. Vector-CD event counts
are not identified with the number of serial single-mode subpulses.
No calibrated value is assigned to $\tau_{\rm unit}$. Non-CD rotation durations, state preparation, readout, and noise remain outside this CD-time comparison.

\paragraph{Reference and outer-step convergence.}
The exact reference has maximum $P_{02}=0.939593$ near resonance and
$0.170523$ in the detuned control over $1$~ps. The near-resonant exchange period is $0.31$~ps. A two-level dyad with the same $|W_{\rm F}|$ and $\delta$ gives $0.32$~ps and a maximum $P_{02}=0.994$. The difference comes from couplings of the bounded interaction outside the dyad.
The $N_F=20\to24$ maximum sampled differences in $P_{02}$ are
$2.26\times10^{-7}$ and $1.54\times10^{-8}$. The corresponding maximum raw
state-vector differences are $4.17\times10^{-4}$ and $9.11\times10^{-5}$.
Cutoffs $12$ and $16$ were also checked.
These are convergence diagnostics, not rigorous continuum error bounds.
We also changed $d_{\rm reg}$ from $0.25$ to $0.10$ and recalibrated the doorway coupling. This alters the near-resonant $P_{02}$ by at most $0.00383$ at $N_F=24$.

An independent exact-$V$ comparison at $\Delta t=2,1,0.5,0.25,0.125$~fs
selects $0.25$~fs as the first tested step with maximum raw state error
below $2\times10^{-3}$ for both detunings. Its errors are
$1.8683\times10^{-3}$ and $1.0940\times10^{-3}$ against the $N_F=24$ exact
reference. A step half as large reduces them to $4.6703\times10^{-4}$ and $2.7348\times10^{-4}$.
\paragraph{Compiled trajectories.}
At $\eps_V=3\times10^{-6}$ the near-resonant compiler-only state error is $2.17\times10^{-4}$ and ancilla leakage is $4.71\times10^{-8}$. Intermediate budgets need not give monotone trajectory error because the complementary blocks interfere coherently.

The least-CD-time candidates meeting the joint-state threshold for both detunings all use $\eps_V=10^{-5}$ in the fixed six-budget scan. Table~\ref{tab:fermi-compilers} reports their separate detuning errors, population errors, and ancilla populations.

\begin{table*}[t]
\caption{Least-CD-time candidates in the fixed six-budget scan with state distance
$e_{\rm joint}\le10^{-2}$, which bounds the joint trace distance $D$, for both detunings. These four selections all use
$\epsilon_V=10^{-5}$. They are not globally optimal schedules.
Errors are maxima on the common $2$~fs output grid.
The $N_F=24$ reference and outer-step errors are distinct diagnostics.}
\label{tab:fermi-compilers}
\begin{ruledtabular}
\begin{tabular}{lrrrrr}
Method & $T_{\rm CD}/\tau_{\rm unit}$ & $e_{\rm joint}$ near
& $e_{\rm joint}$ off & $\max|\Delta P_{02}|$ near
& $\max p_{\rm ancilla=1}$ near\\
\hline
GJA--OQ--QSP & $52\,000$ & $2.736\,10^{-3}$ & $5.446\,10^{-3}$
 & $3.933\,10^{-4}$ & $3.992\,10^{-6}$\\
Direct Fourier & $52\,000$ & $2.736\,10^{-3}$ & $5.445\,10^{-3}$
 & $3.933\,10^{-4}$ & $3.994\,10^{-6}$\\
TGIFS symmetric & $476\,000$ & $2.633\,10^{-3}$ & $3.502\,10^{-3}$
 & $7.124\,10^{-4}$ & $3.603\,10^{-6}$\\
TGIFS basic & $203\,259\,000$ & $2.667\,10^{-3}$ & $4.130\,10^{-3}$
 & $3.938\,10^{-4}$ & $3.624\,10^{-6}$
\end{tabular}
\end{ruledtabular}
\end{table*}

\subsection{Short-window dynamics with coherent ancilla reuse}
\label{app:short-window-dynamics}

This appendix gives the details of the noisy simulation of Sec.~\ref{sec:short-window-dynamics} and Fig.~\ref{fig:fermi-grouped}(c).

\paragraph{Model and reference.}
The reference is the evolution under Eqs.~\eqref{eq:fermi-H} and \eqref{eq:fermi-regularization} with $d_{\rm reg}=0.25$ and the parameters of Appendix~\ref{app:grouped-benchmarks}, without time splitting. The initial state is $\ket0_a\ket{1,0}$ and the window is the first $80$~fs. At $80$~fs the reference populations are $P_{02}=0.481952$ near resonance and $0.153736$ off resonance. The compiled circuits use a smoother potential with $d=0.75$ and
$\kappa(d)=51.232/[\,\{2e^{-d^2/2}-\tfrac12e^{-2d^2}\}/3\,]$.
This choice keeps the matrix element between $\ket{1,0}$ and $\ket{0,2}$ at its reference value (doorway matching) but changes the Hamiltonian. All errors below are measured against the $d_{\rm reg}=0.25$ reference, so they include this change. Exact propagation with the $d=0.75$ potential alone gives $D_{\max}=0.034363$ at $N_F=32$. The synthesis ancilla is never measured, reset, projected, or renormalized.

\paragraph{Accuracy criteria.}
The main error measure is the trace distance $D(t)=\tfrac12\|\rho(t)-\rho_{\rm ref}(t)\|_1$, where $\rho_{\rm ref}$ includes the ancilla in $\ket0$. We take its maximum over the output times $t=0,5,\ldots,80$~fs and over both detunings. For pure states, the state distance of the $1$~ps benchmark is an upper bound on $D$. A circuit is admitted if $D\le0.1$ without noise or $D\le0.2$ with noise. In addition, the errors in $P_{10}$, in $P_{02}$, and in the complex coherence $\rho_{10,02}$, and the ancilla leakage, must each be at most $0.05$. Populations and coherences include both ancilla outcomes.

\paragraph{Circuits.}
The GJA coefficients come from explicit Bessel convolution and agree with an independent Fourier transform. Direct Fourier synthesis would therefore give the same circuits, so this test gives no evidence of a noise advantage of GJA over it. Truncation uses a common global normalization with contractivity margin $10^{-8}$, and all three groups use inside roots.

The rule-based circuit of Fig.~\ref{fig:fermi-grouped}(c) adds the step-alternating completion phase of Sec.~\ref{sec:composition}, with phase $0$ on even steps and $\pi$ on odd steps. It involves no search over roots or relative phases. With this rule, $(K_+,K_-,K_s)=(3,3,5)$ is the lowest-cost admitted allocation among all allocations with $K_\pm\le4$ and $K_s\le8$. Its noiseless maximum $D$ is $0.0912$ at $N_F=24$, $32$, and $48$ for both detunings. The nearest cheaper allocation, $(3,3,4)$, fails with $D=0.1023$, and $(3,3,5)$ without the alternation fails with $D=0.1119$. The baseline circuit uses $(4,4,5)$ without completion phases and gives near/off ideal $D=0.088971/0.081101$ at $N_F=48$. The admission of the rule-based circuit is specific to the input state. The holdout inputs $\ket{0,2}$ and $(\ket{1,0}+\ii\ket{0,2})/\sqrt2$ give $D=0.1162$ and $0.1044$.

\paragraph{Pulse and noise model.}
Each vector query is resolved into consecutive one-mode pulses $\exp(\ii\nu QZ/2)$. We assign $t_{\rm CD}=0.028+0.492|\nu|/(\pi/2)$~$\mu$s to each one-mode pulse and $28$~ns to each nontrivial qubit rotation, including step-boundary rotations. The $520$~ns ECD and $28$~ns rotation times in Ref.~\cite{Fontaine2026Programming} motivate the scale. The amplitude-dependent clock is our modeling assumption, and its conversion to calibrated runtimes remains open. Harmonic gates are zero-duration virtual updates, with no additional harmonic drift during pulses. The rule-based circuit then contains $544$ modewise CDs and $368$ rotations, with a modeled duration of $153.329~\mu\mathrm{s}$.

A local Pauli-frame rewrite, $R\mapsto g_{\rm after}Rg_{\rm before}^\dagger$ with $\nu\mapsto-\nu$ in an $\ii X$ frame and identity boundary frames, leaves each complete ideal word unchanged. The frames are fixed offline from the mean near/off excited-state population, preserve the pulse count and degree, and are frozen across noise rates and cutoffs. Both circuits receive the same frame-selection algorithm, and signed CDs are assumed available. In the replay check, the rewritten words agree with the ideal words to a projective residual below $5\times10^{-15}$, and the native replay residual of the rule-based circuit is below $1.2\times10^{-14}$.

Pulsewise Lindblad evolution applies the effective pulse Hamiltonian and the dissipators simultaneously. It uses
$L_-=\sqrt{\eta_{\rm n}/T_1}\sigma_-$,
$L_\phi=\sqrt{\eta_{\rm n}/(2T_\phi)}Z$, and
$L_j=\sqrt{\eta_{\rm n}/T_{\rm cav}}a_j$,
with the illustrative lifetimes $T_1=T_\phi=60$~$\mu$s and $T_{\rm cav}=160$~$\mu$s. Equation~\eqref{eq:short-window-fidelity} is evaluated in the squared convention directly from the saved final density matrices and the unsplit pure reference, including the ancilla in $\ket0$.

\paragraph{Results.}
Table~\ref{tab:staged-ablation} compares the baseline and the rule-based circuits.

\begin{table}[t]
\caption{Near-resonant comparison of the baseline and rule-based circuits of Sec.~\ref{sec:short-window-dynamics}. Both circuits use the common Pauli-frame rule and pulse clock. Fidelities use $N_F=24$ and the common joint target at $80$~fs.}
\label{tab:staged-ablation}
\begin{ruledtabular}
\begin{tabular}{lrr}
Quantity & Baseline & Rule\\
\hline
$(K_+,K_-,K_s)$ & $(4,4,5)$ & $(3,3,5)$\\
Modewise CDs & 672 & 544\\
Rotations & 432 & 368\\
Duration ($\mu$s) & 188.77 & 153.33\\
$F$, $\eta_{\rm n}=0$ & 0.992083 & 0.995485\\
$F$, $\eta_{\rm n}=0.03$ & 0.928867 & 0.933148\\
$F$, $\eta_{\rm n}=0.1$ & 0.797641 & 0.803698\\
$F$, $\eta_{\rm n}=1$ & 0.158612 & 0.168777\\
\end{tabular}
\end{ruledtabular}
\end{table}

At $\eta_{\rm n}=0.03$, the near-resonant coherence errors are $0.052139$ for the baseline and $0.050042$ for the rule-based circuit, both above the threshold $0.05$. The other near-resonant criteria pass for both circuits, and the off-resonant case of the rule-based circuit passes all four criteria.

\paragraph{Convergence with the Fock cutoff.}
Cutoff convergence is reported as the mixed-state difference $D$ between two cutoffs on the $5$~fs output grid, with target $0.005$. For the rule-based circuit at $\eta_{\rm n}=0.03$, the maximum difference between $N_F=24$ and $32$ is $0.0030/0.0036$ (near/off), below the target, and the final fidelity changes by at most $6.2\times10^{-7}$. For the baseline, the same difference is $0.006542$, above the target, while its final target fidelity changes by at most $1.2\times10^{-6}$ and its coherence failure persists at $N_F=32$. The ideal native-prefix differences of the baseline are $0.01031$ from $24$ to $32$, $0.00505$ from $32$ to $48$, and $0.000960$ from $48$ to $64$. These finite-cutoff checks distinguish output-grid agreement from convergence throughout each pulse.

\par
\renewcommand{\bibfont}{\footnotesize}
\bibliographystyle{apsrev4-2}
\setlength{\bibsep}{-1pt}
\bibliography{ja_oqgqsp_refs}

\end{document}